\documentclass[11pt]{article}
\usepackage[T2A]{fontenc}
\usepackage{ alphalph,etoolbox }
\usepackage{amsthm, amsmath, amssymb, amsbsy, amscd, amsfonts, latexsym, euscript,epsf, bm}
\newtheorem{thm}{Theorem}[section]
\usepackage{bbold, bm, stackrel}
\usepackage{epic,eepic}
\usepackage{texdraw}
\usepackage[dvips]{color}
\usepackage{color}
\usepackage{ dsfont }
\usepackage{graphicx}
\usepackage{exscale,relsize}
\usepackage{authblk}
\usepackage{url}
\usepackage{hyperref}
\usepackage{enumerate}
\usepackage{graphicx}
\newtheorem{proposition}[thm]{Proposition}
\newtheorem{corollary}[thm]{Corollary}

\newtheorem{theorem}[thm]{Theorem}

\theoremstyle{definition}

\newtheorem{remark}[thm]{Remark}

\DeclareMathOperator{\rank}{rank}

\DeclareMathOperator{\sech}{sech}
\DeclareMathOperator{\diag}{diag}

\title{Noncommutative NLS systems: Darboux--B\"acklund transformations and integrable discretisations}
\date{}

\author{S. Konstantinou-Rizos\thanks{skonstantin84@gmail.com}}
\affil{Centre of Integrable Systems, P.G. Demidov Yaroslavl State University, Yaroslavl, Russia}

\patchcmd{\subequations}{\alph{equation}}{\alphalph{\value{equation}}}{}{}

\begin{document}

\maketitle

\begin{abstract}
We study noncommutative analogues and integrable discretisations of nonlinear Schr\"odinger (NLS)-type systems associated with reduction groups. In particular, we consider the Ablowitz--Kaup--Newell--Segur (AKNS) system, the Kaup--Newell derivative NLS system, and the Mikhailov--Shabat--Yamilov deformation of the derivative NLS system together with their Darboux--Bäcklund transformations and associated lattice equations.

We derive the continuum limits of previously constructed integrable lattice systems and recover the corresponding NLS-type partial differential equations. We then construct a noncommutative deformation of the Mikhailov--Shabat--Yamilov system and show that, unlike the AKNS and Kaup--Newell cases, its Lax representation requires the introduction of nonlocal variables.

Furthermore, we derive Darboux--B\"acklund transformations and integrable discretisations for the noncommutative derivative NLS and deformation derivative NLS systems in the form of vertex--bond lattice equations. We also construct explicit solutions for a six-point derivative NLS-type lattice equation and for its noncommutative analogue.
\end{abstract}

\bigskip


\begin{quotation}
\noindent{\bf PACS numbers:}
02.30.Ik, 02.90.+p, 03.65.Fd.
\end{quotation}

\begin{quotation}
\noindent{\bf Mathematics Subject Classification 2020:}
35Q55, 16T25.
\end{quotation}

\begin{quotation}
\noindent{\bf Keywords:} Noncommutative Darboux transformations, noncommutative B\"acklund transformations, NLS-type integrable PDEs, NLS-type integrable lattice equations.
\end{quotation}

\section{Introduction}\label{intro}
Darboux and B\"acklund transformations play an important role in theory of integrable systems in mathematical physics. They constitute mechanisms for constructing solutions to partial differential equations (PDEs) starting from trivial ones. They can also be employed to linearise nonlinear PDEs and serve as generators of Yang--Baxter and Zamolodchikov tetrahedron maps \cite{BIKRP, Sokor-Sasha, Sokor-Pap, Sokor-2020-2, MPW}. They are equally useful for the construction of soliton solutions for partial difference (lattice) equations \cite{Fisenko-Sokor}.

Undoubtedly, the nonlinear Schr\"odinger (NLS) type systems are among the most important nonlinear systems of PDEs, playing a fundamental role in mathematical physics, with applications ranging from nonlinear optics and gravity waves to solitons and the theory of integrable systems. The most famous NLS-type systems are probably the Ablowitz--Kaup--Newell--Segur nonlinear Schr\"odinger system \cite{AKNS}, the Kaup--Newell derivative NLS system (DNLS) \cite{Kaup-Newell}, and  the Mikhailov--Shabat--Yamilov deformation of the DNLS system \cite{MSY}. These three NLS-type equations are associated with some classification results of $\mathfrak{sl}(2,\mathbb{C})$-based automorphic Lie algebras corresponding to finite reduction groups \cite{Sasha-Rhys}.

On the other hand, the study of noncommutative integrable PDEs and their associated integrable lattice discretisations is an important and fast-growing field of mathematical physics (we indicatively refer to some recent results \cite{Adler-P, Bobrova, Novikov-Wang,  Xenitidis-NC}). Moreover, noncommutative versions of NLS-type equations have been studied in the literature since the late 90s. We refer, in particular, to a classical result by Olver and Sokolov \cite{OS}, the study of symmetries and conservation laws for nonabelian evolution equations \cite{Adler-Sokolov}, some recent results on noncommutative discretisations of NLS systems \cite{KRX} and the study of Darboux transformations and semi-discretisations of DNLS-type equations \cite{Peroni-Wang}. Other famous results on noncommutative NLS-type systems and their solutions include \cite{Dimakis-Hoissen, Fordy-Kulish, GGI}; also see the books \cite{Kupershmidt, Sokolov} and the references therein. 

The paper aims to extend the results of papers \cite{SPS, KRX}. In particular, in \cite{SPS} we studied the Darboux--B\"acklund transformations and the discretisations of the AKNS, the derivative NLS, and the deformation of the derivative NLS systems. In \cite{KRX} we derived noncommutative versions of the Darboux--B\"acklund transformations of the AKNS system that were presented in \cite{SPS}, and also its noncommutative integrable discretisations, namely the noncommutative Adler system and the noncommutative discrete Toda equation. Furthermore, in \cite{Sokor-Nikitina}, we constructed a noncommutative Darboux-B\"acklund transformation for the derivative NLS system (such transformations were also constructed in \cite{Peroni-Wang} for a family of derivative NLS  systems). Here, we derive noncommutative integrable discretisations of the DNLS and the deformation of the DNLS systems.

The main new results of this paper are as follows:
\begin{itemize}
    \item the derivation of continuum limits from NLS-type discrete lattice systems to the AKNS system, the derivative NLS system, and a  deformation  of  the derivative NLS system;
    
    \item the construction of a noncommutative version of the deformation of the derivative NLS equation;

    \item an approach for constructing nontrivial solutions to such equations;
    
    \item the construction of Darboux--B\"acklund transformations for the noncommutative deformation of the derivative NLS  equation;
    \item noncommutative integrable discretisations of the derivative NLS (Kaup--Newell) and the deformation of the derivative NLS (Mikhailov--Shabat--Yamilov) systems;
    \item explicit wave solutions for a six-point derivative NLS-type lattice equation, and a solution for its noncommutative analogue.
\end{itemize}

The remainder of the paper is organised as follows.

In the next section, we explain the notation used throughout the text and give all the basic definitions that are needed in the following sections. In particular, we explain what equations on quad-graphs are and what integrability in terms of a Lax pair means. Moreover, we explain the notions of an integrable evolution type PDE, the Darboux and B\"acklund transformations. Then, we explain how Darboux transformations can be used to discretise integrable PDEs.

In Section \ref{int-discretisations}, we recall the AKNS, the Kaup--Newell, and the Mikhailov--Shabat--Yamilov integrable systems, their Lax pairs, their Darboux transformations, and their associated discretisations, namely NLS  type integrable lattice systems. We derive the continuum limits of Adler--Yamilov, the derivative NLS-type and a deformation of the derivative NLS-type lattice systems to their continuous origins, namely the  AKNS, the Kaup--Newell, and the Mikhailov--Shabat--Yamilov systems, respectively.

In Section \ref{NC-NLS-eqs}, we recall the noncommutative versions of the AKNS and the Kaup--Newell systems \cite{OS}, and we construct a noncommutative version of the Mikhailov--Shabat--Yamilov system. We present an approach for constructing solutions to such noncommutative  systems.

In Section \ref{NCDBT}, we recall the Darboux matrices and the $x$-parts of the B\"acklund transformations for the noncommutative AKNS and Kaup--Newell systems, and derive the associated $t$-parts of the B\"acklund transformations, which were not derived in \cite{KRX, Sokor-Nikitina}. Since the noncommutative AKNS and Kaup--Newell systems are evolutionary, in principle, the systems themselves play the role of the $t$-part of the B\"acklund transformation. However, since their Darboux matrices contain auxiliary functions, the evolution equations for these auxiliary functions must also be determined. Then, we derive a Darboux and a B\"acklund transformation for the noncommutative Mikhailov--Shabat--Yamilov system.

Section \ref{NCDISNLS} deals with the integrable discretisations of the noncommutative NLS-type systems discussed in section \ref{NC-NLS-eqs}. We first recall the noncommutative discretisations of the AKNS system, and then we construct noncommutative integrable discretisations of the Kaup--Newell and the Mikhailov--Shabat--Yamilov systems in the form of noncommutative vertex-bond systems. For the Kaup--Newell system  we also construct a noncommutative 6-point lattice equation. We also construct wave solutions of the commutative version of this 6-point equation.

Finally, in Section \ref{conclusions}, we summarise the main results and propose several  directions for future work.

\section{Preliminaries}
\subsection{Notation}
Throughout the text:
\begin{itemize}
    \item By $\mathcal{X}$ we denote an arbitrary set, whereas by Latin italic letters (i.e. $x, y, u, v$ etc.) the elements of $\mathcal{X}$, with the exception of the `spectral parameter' which is denoted by the Greek letter $\lambda$. Moreover, by $\mathcal{X}^n$ we denote the Cartesian product of  $\mathcal{X}$ with itself $n$ times, i.e.
$\mathcal{X}^n=\underbrace{\mathcal{X}\times \mathcal{X}\times\dots\times \mathcal{X}}_{n}$.
    
    \item By $\mathfrak{R}$ we denote a noncommutative division ring, and its elements are denoted by bold italic Latin letters (i.e. $\bm{x}, \bm{y}, \bm{u}$ etc.). That is, $\mathfrak{R}$ is a ring with multiplicative identity $1$ where commutativity with respect to multiplication is not assumed, and every nonzero element $\bm{x}$ has an inverse $\bm{x}^{-1}$, i.e. $\bm{x}\bm{x}^{-1}=\bm{x}^{-1}\bm{x}=1$. The centre of a division ring will be denoted by $Z(\mathfrak{R})=\{a\in\mathfrak{R}:\forall\bm{x}\in\mathfrak{R},a\bm{x}=\bm{x}a\}$.
    
    \item Matrices will be denoted by uppercase Roman letters (i.e. ${\rm A}, {\rm B}, {\rm C}$) etc. Additionally, matrix operators are denoted by capital calligraphic letters (for instance, $\mathcal{L}=D_x-\rm{U}$). By $\sigma_1$ and $\sigma_3$ we denote the Pauli matrices $\sigma_1=\begin{pmatrix} 0 &  1\\ 1  & 0\end{pmatrix}$ and $\sigma_3=\begin{pmatrix} 1 & 0\\ 0 &-1 \end{pmatrix}$.
\end{itemize}

\subsection{Quad-graph integrable systems}
Let $\mathcal{X}$ be a set of variables and $S$ a set of parameters. A quad-graph equation (or system of equations) is  a polynomial equation  of the form
\begin{equation}\label{quad-graph}
    Q(f_{00},f_{10},f_{01},f_{11},a,b)=0,
\end{equation}
where $f_{ij}$, $i,j\in\{0,1\}\in\mathcal{X}$, and $a,b\in S$. The (nonlinear) function $Q$  should depend on all the $f_{ij}$, $i,j\in\{0,1\}$, and should be affine linear, namely
$$
\frac{\partial Q}{\partial f_{ij}}\neq 0,\qquad \frac{\partial^2 Q}{\partial f_{ij}^2}=0.
$$

Equation \eqref{quad-graph} is called a quad-graph equation because it admits a quad representation as in Figure \ref{QG}, where the fields $f_{ij}$, $i,j\in\{0,1\}$, lie on the vertices of the square, while $a$ and $b$ lie on the  edges.

\begin{figure}[ht]
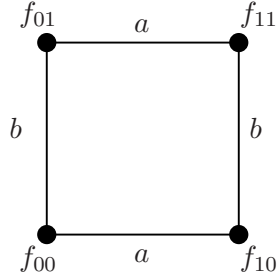

\centertexdraw{ \setunitscale 0.5
\linewd 0.02 \arrowheadtype t:F 
\htext(0 0.5) {\phantom{T}}
\move (-1 -2) \lvec (1 -2) 
\move(-1 -2) \lvec (-1 0) \move(1 -2) \lvec (1 0) \move(-1 0) \lvec(1 0)
\move (1 -2) \fcir f:0.0 r:0.1 \move (-1 -2) \fcir f:0.0 r:0.1
 \move (-1 0) \fcir f:0.0 r:0.1 \move (1 0) \fcir f:0.0 r:0.1  
\htext (-1.3 -2.4) {$f_{00}$} \htext (1 -2.4) {$f_{10}$} \htext (-0.1 -2.3) {$a$}
\htext (-1.3 .15) {$f_{01}$} \htext (1 .15) {$f_{11}$} \htext (-0.1 .1) {$a$}
\htext (-1.4 -1) {$b$} \htext (1.1 -1) {$b$}}
\caption{{Quad-graph equation}} \label{QG}
\end{figure}

Now, equation \eqref{quad-graph} can be understood as a system of  partial difference equations (P$\Delta$E) in the following way: consider $f$ to be a function of two discrete variables $n,m\in\mathbb{Z}$, and let the subscripts in $f_{ij}$ indicate shifts  with respect to $n$ and $m$,  respectively, namely $f_{ij}=f(n+i,m+j)$,  $f_{00}\equiv f(n,m)$. Then, the square in Figure \ref{QG} can be interpreted as an elementary square of a two-dimensional lattice.
Among the most famous examples of quad-graph equations is the discrete potential KdV (dpKdV) equation 
\begin{equation}\label{dpKdV}
   (f_{00}-f_{11})(f_{01}-f_{10})=a-b.
\end{equation}

One definition for integrability of a quad-graph equation is the existence of a Lax pair. Specifically, a quad-graph equation \eqref{quad-graph} is  called integrable if it can be written as
\begin{equation}\label{d-Lax}
    {\rm L}(f_{01},f_{11},a,\lambda){\rm M}(f_{00},f_{01},b,\lambda)={\rm M}(f_{10},f_{11},b,\lambda){\rm L}(f_{00},f_{10},a,\lambda),\quad \text{for any}~~\lambda\in\mathbb{C},
\end{equation}
for some matrices ${\rm L}={\rm L}(f_{00},f_{10},a,\lambda)$  and ${\rm M}={\rm M}(f_{00},f_{01},b,\lambda)$. This is the discrete  zero-curvature condition, and it is called a Lax representation for the associated quad-graph equation \eqref{quad-graph}.

\subsection{Darboux--B\"acklund transformations and discretisation of integrable PDEs}
Let $u=u(x,t)$ be a solution of an evolution type PDE
\begin{equation}\label{evol-eq}
    u_t=f(u,u_x,u_{xx},\ldots;x,t),
\end{equation}
where indices denote  partial derivatives $u_t=\partial_t u(x,t),~u_x=\partial_xu(x,t),~u_{xx}=\partial^2_xu(x,t),...$ etc. Equation \eqref{evol-eq} is called integrable if there exist operators $(\mathcal{L},\mathcal{A})\equiv (D_x-\rm{U}(u,\lambda),D_t-\rm{V}(u,\lambda))$, such that equation \eqref{evol-eq} can be written as the \textit{zero-curvature condition}
\begin{equation}\label{zero-curvature}
    [\mathcal{L},\mathcal{A}]=0 ~~\Leftrightarrow ~~ \rm{U}_t-\rm{V}_x+[\rm{U},\rm{V}]=0.
\end{equation}
The pair of operators $(\mathcal{L},\mathcal{A})$ (or the pair of matrices $(\rm{U},\rm{V})$) is called a Lax pair for equation \eqref{evol-eq}.

A Darboux transformation is the following transformation that leaves $\mathcal{L}$ and $\mathcal{A}$ covariant, defined by
$$
\mathcal{L}(u,\lambda)\rightarrow {\rm M}\mathcal{L}(u,\lambda){\rm M}^{-1}=\mathcal{L}(\tilde{u},\lambda),\qquad \mathcal{A}(u,\lambda)\rightarrow {\rm M}\mathcal{A}(u,\lambda){\rm M}^{-1}=\mathcal{A}(\tilde{u},\lambda).
$$
Thus, the Darboux transformation maps a solution, $u(x,t)$, of \eqref{evol-eq} to another solution, $\tilde{u}(x,t)$.

The matrix $\rm{M}$ in the above definition is called a Darboux matrix and, by definition, satisfies the system of differential equations
\begin{equation}\label{M-equation}
    \rm{M}_x+{\rm M}{\rm U}-\tilde{{\rm U}}{\rm M}=0,\qquad \rm{M}_t+{\rm M}{\rm V}-\tilde{{\rm V}}{\rm M}=0,
\end{equation}
where $\tilde{\rm{U}}=\rm{U}(\tilde{u},\lambda)$ and $\tilde{\rm{V}}=\rm{V}(\tilde{u},\lambda)$. Thus, a Darboux transformation consists of a Darboux matrix
\begin{equation}\label{DM-f}
    {\rm M}={\rm M}(u,\tilde{u},f,\lambda),
\end{equation}
where $f$ is an arbitrary function, and a system of differential equations for $u$, $\tilde{u}$ and $f$, namely a B\"acklund transformation. If, additionally, matrices ${\rm U}$ and ${\rm V}$ belong to the algebra $\mathfrak{sl}_2(\mathbb{C})$, then it follows that the determinant of the Darboux matrix is constant (Abel--Liouville theorem). The constant determinant can be used to express the function $f$ in terms of $u$, $\tilde{u}$, and a parameter $a$. That is, the Darboux matrix takes the form:
\begin{equation}\label{DM-a}
    {\rm M}={\rm M}(u,\tilde{u},a,\lambda),
\end{equation}
Equivalently, $\det(\rm M)$ is a first integral of the B\"acklund transformation.

Now, once a Darboux matrix ${\rm M}(u,\tilde{u},a,\lambda)$ is constructed, a second one can be obtained by changing $a\rightarrow b$ and $\tilde{u}\rightarrow \hat{u}$, i.e. ${\rm K}={\rm M}(u,\hat{u},b,\lambda)$. Having two Darboux matrices, ${\rm M}$ and ${\rm K}$, at our disposal, we can discretise equation \eqref{evol-eq} in the following way: assume that function $u$ depends also on two discrete variables $n,m\in\mathbb{N}$, namely $u=u(x,t;n,m)$. Then, we use these variables to place various solutions $u$ of \eqref{evol-eq} on a two-dimensional lattice, by setting $u=u_{00}=u(x,t;n,m)$, $\tilde{u}=u_{10}=u(x,t;n+1,m)$, $\hat{u}=u_{01}=u(x,t;n,m+1)$, adopting the notation introduced earlier.

Now, in the new notation, the Darboux matrices take the form ${\rm M}={\rm M}(u_{00},u_{10},a,\lambda)$ and ${\rm K}={\rm M}(u_{00},u_{01},b,\lambda)$. We substitute them to the discrete zero-curvature condition \eqref{d-Lax}:
\begin{equation}\label{ZC-ab}
    {\rm M}(u_{01},u_{11},a,\lambda){\rm M}(u_{00},u_{01},b,\lambda)={\rm M}(u_{10},u_{11},b,\lambda){\rm M}(u_{00},u_{10},a,\lambda).
\end{equation}
If the above equation is equivalent to a system of polynomial equations $Q(u_{00},u_{10},u_{01},u_{11};a,b)$, for any $\lambda\in\mathbb{C}$, then the latter is integrable by construction, since it has a Lax pair $({\rm M}, {\rm K})$.

\begin{remark}\normalfont
    We can use the version \eqref{DM-a} of the Darboux matrix with the parameter and obtain a quad-graph system from \eqref{ZC-ab}. Alternatively, we may use the version \eqref{DM-f} of the Darboux matrix, before writing $f$ in terms of $u$, $\tilde{u}$ and $a$, namely use matrix ${\rm M}={\rm M}(u_{00},u_{10},f_{00})$. Then, obtain a second matrix ${\rm K}$ by changing $f_{00}\rightarrow g_{00}$ and $u_{10}\rightarrow u_{01}$ in ${\rm M}$, namely ${\rm K}={\rm M}(u_{00},u_{01},g_{00},\lambda)$, and consider the zero-curvature equation  
    $$
    {\rm M}(u_{01},u_{11},f_{01},\lambda){\rm M}(u_{00},g_{00},\lambda)={\rm M}(u_{10},u_{11},g_{10},\lambda){\rm M}(u_{00},u_{10},f_{00},\lambda).
    $$
    If the above Lax equation is equivalent to a system of polynomial equations, then the latter is integrable by construction. This system can be interpreted on a square with $u_{00},u_{10},u_{01},u_{11}$ lying on its vertices while $f_{00},f_{01}$ and $g_{00}, g_{10}$ on its horizontal and vertical edges, respectively. This is called a vertex-bond type system \cite{Hiet-Viallet} and, by construction, it can be restricted to the quad-graph system equivalent to \eqref{ZC-ab} on the level sets of its first integrals.
\end{remark}

\section{Integrable discretisations of NLS-type systems}\label{int-discretisations}
In \cite{SPS} we presented integrable discretisations of certain NLS-type  systems, namely the AKNS, the DNLS and the deformation of the DNLS systems. In this section, following the continuum limit approach \cite{Hiet-Frank-Joshi}, we present continuum limits connecting the obtained integrable systems of partial difference equations (P$\Delta$Es) to the  original NLS-type systems of PDEs.

\subsection{Ablowitz--Kaup--Newell--Segur system} 
The Ablowitz--Kaup--Newell--Segur (AKNS) nonlinear Schr\"odinger system reads:
\begin{equation}\label{NLS-eq}
p_t=\frac{1}{2}p_{xx}-4p^2q, ~~~ q_t=-\frac{1}{2}q_{xx}+4q^2p,
\end{equation}
where $x\in\mathbb{R}$, $t>0$, and it is equivalent to the zero-curvature condition \eqref{zero-curvature} for
\begin{equation}\label{Lax-NLS}
    {\rm{U}}=\lambda\begin{pmatrix}
   1 & 0\\
   0 & -1
\end{pmatrix}
+
\begin{pmatrix}
   0 & 2p\\
   2q & 0
\end{pmatrix}, \quad
 {\rm{V}}=\lambda^2 \begin{pmatrix}
   1 & 0\\
   0 & -1
\end{pmatrix}
+\lambda \begin{pmatrix}
   0 & 2p\\
   2q & 0
\end{pmatrix}
+
\begin{pmatrix}
   -2pq & p_x\\
   -q_x & 2pq
\end{pmatrix}.
\end{equation} 

System \eqref{NLS-eq} admits a Darboux transformation that consists of the Darboux matrix
\begin{subequations} \label{Darboux-NLS}
\begin{equation}\label{Darboux-mat-NLS}
{\rm{M}}(p,\tilde{q},a)=\lambda\begin{pmatrix}
1 & 0 \\
0 & 0
\end{pmatrix}
+
\begin{pmatrix}
a+pq & p \\
\tilde{q} & 1
\end{pmatrix},
\end{equation}
and the associated $x$-part of the B\"acklund transformation:
\begin{equation}\label{Backlund-NLS}
p_x=2\tilde{p} -2(a+p\tilde{q})p,\quad q_x=2q(a+p\tilde{q})-2q.
\end{equation}
\end{subequations}

For the discretisation of system \eqref{NLS-eq},  in \cite{SPS} we employed matrix \eqref{Darboux-mat-NLS} in the discrete setting explained in the previous section, and define matrices
\begin{equation}
    {\rm M}_{00}={\rm{M}}(p_{00},q_{10},a)=\begin{pmatrix}
\lambda+a+p_{00}q_{10} & p_{00} \\
q_{10} & 1
\end{pmatrix},\quad {\rm K}_{00}={\rm{M}}(p_{00},q_{01},b)=\begin{pmatrix}
\lambda+b+p_{00}q_{01} & p_{00} \\
q_{01} & 1
\end{pmatrix},
\end{equation}
and substitute them into the Lax equation ${\rm M}_{01}{\rm K}_{00}={\rm K}_{10}{\rm M}_{00}$, which is equivalent to the Adler--Yamilov system \cite{SPS}
\begin{equation}\label{Adler-Yamilov}
    p_{01} =p_{10}-\frac{a-b}{1+p_{00}q_{11}}p_{00},\quad q_{01} =q_{10}+\frac{a-b}{1+p_{00}q_{11}}q_{11},
\end{equation}
for all $\lambda\in\mathbb{C}$.

For the continuum limit, we introduce a small parameter $h \to 0$ and define 
$$
a = 1 + h, \quad b = 1 - h.
$$
Then, we define slow variables
\begin{equation}\label{slow-vars}
    x = h(n+m) - h^2(n-m), \quad t = h^3(n-m).
\end{equation}
We also scale the dependent variables:
\begin{equation}\label{potentials}
    p_{n,m} = 2h a^n b^m p(x,t), \quad q_{n,m} = 2h a^{-n} b^{-m} q(x,t).
\end{equation}
Next, we introduce the discrete notation $x_{ij}=x(n+i,m+j)$ and $t_{ij}=t(n+i,m+j)$, which implies
\begin{subequations}\label{xij-tij}
    \begin{align}
        x_{00}=x,\quad x_{10}&=x+h-h^2,\quad x_{01}=x+h+h^2,\quad x_{11}=x+2h-h^2,\label{xij-tij-a}\\
        t_{00}&=t,\quad t_{10}=t+h^3,\quad t_{01}=t-h^3,\quad t_{11}=t.\label{xij-tij-b}
    \end{align}
\end{subequations}

We have the following.

\begin{proposition}\label{AY-lim}
    The continuum limit of the Adler--Yamilov system \eqref{Adler-Yamilov} is the AKNS system \eqref{NLS-eq}.
\end{proposition}
\begin{proof}
    Let $p_{i,j}=2ha^{n+i}b^{m+j}p(x_{ij},t_{ij})$ and $q_{n,m}=2ha^{-n-i}b^{-m-j}q(x_{ij},t_{ij})$. Then,
    $$
    p_{00}=2ha^{n}b^{m}p(x_{00},t_{00}),\quad p_{10}=2ha^{n+1}b^{m}p(x_{10},t_{10}),\quad p_{01}=2ha^{n}b^{m+1}p(x_{01},t_{01}),
    $$
    and
    $$
    q_{10}=2ha^{-n-1}b^{m}q(x_{10},t_{10}),\quad q_{01}=2ha^{-n}b^{-m-1}q(x_{01},t_{01}),\quad q_{11}=2ha^{-n-1}b^{-m-1}q(x_{11},t_{11}).
    $$
    Now, substituting $x_{ij}$ from \eqref{xij-tij} into the above, and expanding in Taylor series up to $\mathcal{O}(h^5)$, we obtain that
    \begin{align*}
         p_{10} &= 2h a^nb^m \left[(1+h)\left(p + (h-h^2)p_x + h^3 p_t + \frac{h^2}{2}p_{xx}+\mathcal{O}(h^3)\right)\right],\\ 
         p_{01} &= 2h a^nb^m \left[(1-h)\left(p + (h+h^2)p_x - h^3 p_t + \frac{h^2}{2}p_{xx}+\mathcal{O}(h^3)\right)\right],
    \end{align*}
    from where it follows that
     \begin{equation}\label{AY-left}
     p_{10}-p_{01} = 4 p h^2 + (4p_t - 2p_{xx}) h^4 + \mathcal{O}(h^5). 
     \end{equation}
     
     On the other hand, we have that $a-b=2h$, and 
     $$
     p_{00}q_{11}=(2h p)(2h q(x+2h,t))=4h^2 pq + \mathcal{O}(h^3).
     $$ 
     Thus, 
     $$
     \frac{1}{1+p_{00}q_{11}} = 1 + 4 pq h^2 + \mathcal{O}(h^3),
     $$
    which, consequently, implies that
\begin{equation}\label{AY-right}
    \frac{a-b}{1+p_{00}q_{11}}p_{00} = 4ph^2 - 16 p^2 q h^4 + \mathcal{O}(h^5).
\end{equation}
Comparing the coefficients of $h^4$ in \eqref{AY-left} and \eqref{AY-right} we obtain $p_t = \frac{1}{2}p_{xx} - 4p^2 q$, which is the first equation of the AKNS system \eqref{NLS-eq}. Similarly, it can be shown that from the coefficients of $h^4$ of the expansions $q_{01}-q_{10}$ and $ \frac{a-b}{1+p_{00}q_{11}}q_{11}$, we obtain the second equation of the AKNS system \eqref{NLS-eq}.
\end{proof}

\subsection{Kaup--Newell system} 
The Kaup--Newell derivative NLS (DNLS) system \cite{Kaup-Newell} is given by
\begin{equation}\label{DNLS-system}
        p_t = \frac{1}{2} p_{xx}-2\,(p^2q)_x,\quad q_t = -\frac{1}{2}q_{xx}-2\,(q^2p)_x,
\end{equation}
where $x\in\mathbb{R}$, $t>0$, and it is equivalent to equation \eqref{zero-curvature}, where the matrices ${\rm U}$ and ${\rm V}$ are given as follows:
\begin{subequations}\label{Lax-DNLS}
    \begin{align}
          {\rm U}&=\lambda^2\begin{pmatrix}
1 & 0\\
0 & -1
\end{pmatrix}
+
\lambda
\begin{pmatrix}
0 & 2p\\
2q & 0
\end{pmatrix},\label{Lax-DNLS-a}
\\
{\rm V}&=
\lambda^4
\begin{pmatrix}
1 & 0\\
0 & -1
\end{pmatrix}
+
\lambda^3
\begin{pmatrix}
0 & 2p\\
2q & 0
\end{pmatrix}
-2\lambda^2
\begin{pmatrix}
pq & 0\\
0 & -pq
\end{pmatrix}
+
\lambda
\begin{pmatrix}
0 & p_x-4p^2 q\\
-q_x-4pq^2 & 0
\end{pmatrix}.\label{Lax-DNLS-b}  
    \end{align}
\end{subequations}
In this case, the matrices ${\rm U}$ and ${\rm V}$ are invariant under the transformation
\begin{equation}\label{Z2-symmetry}
    \mathfrak{s}_1(\lambda):~{\rm U}(\lambda)\rightarrow \sigma_3{\rm U}(-\lambda)\sigma_3,\quad \mathfrak{s}_1(\lambda):~{\rm V}(\lambda)\rightarrow \sigma_3{\rm V}(-\lambda)\sigma_3.
\end{equation}
The above involution generates the reduction group, which is isomorphic to $\mathbb{Z}_2$.

In \cite{SPS} we constructed a Darboux transformation for system \eqref{DNLS-system} that consists of the Darboux matrix
\begin{subequations}\label{Kaup-Newell-Darboux}
\begin{equation}\label{Kaup-Newell-DT}
{\rm M} =\lambda^2\begin{pmatrix} f & 0\\ 0 & 0\end{pmatrix}+\lambda \begin{pmatrix} 0 & fp\\ fq & 0 \end{pmatrix}+\begin{pmatrix} c_1 & 0\\ 0 & c_2\end{pmatrix}
\end{equation}
and the associated $x$-part of the B\"acklund transformation:
\begin{align}\label{Kaup-Newell-BT}
f_x=2f(\tilde{p}\tilde{q}-pq),\quad
    (fp)_x=2c_2\tilde{p}-2c_1p,\quad
    (f\tilde{q})_x=2c_1\tilde{q}-2c_2q.
\end{align}
\end{subequations}
Moreover, we showed that, for $c_1=c_2=0$, the B\"acklund transformation \eqref{Kaup-Newell-BT} can be integrated, and we obtain a degenerate Darboux transformation that consists of 
\begin{subequations}\label{Kaup-Newell-Darboux-deg}
\begin{equation}\label{Kaup-Newell-DT-deg}
{\rm M} =\lambda^2\begin{pmatrix} \frac{1}{p} & 0\\ 0 & 0\end{pmatrix}+\lambda \begin{pmatrix} 0 & 1\\ 1 & 0 \end{pmatrix},
\end{equation}
together with the associated $x$-part of the B\"acklund transformation:
\begin{align}\label{Kaup-Newell-BT-deg}
\tilde{q}=p,\quad \partial_x p=2p^2 (\tilde{p}-q).
\end{align}
\end{subequations}

One discretisation is obtained by employing matrix \eqref{Kaup-Newell-DT} and consider matrices
\begin{equation}\label{Kaup-Newell}
    {\rm M}_{00}={\rm M}(f_{00},p_{00},q_{10})=\begin{pmatrix} \lambda^2 f_{00}+1 & \lambda f_{00}p_{00}\\ \lambda f_{00}q_{10} & 1 \end{pmatrix},\quad {\rm K}_{00}={\rm M}(g_{00},p_{00},q_{01})=\begin{pmatrix} \lambda^2 g_{00}+1 & \lambda g_{00}p_{00}\\ \lambda g_{00}q_{01} & 1 \end{pmatrix},
\end{equation}
where we have set $c_1=c_2=1$, and substitute them into the Lax equation ${\rm M}_{01}{\rm K}_{00}={\rm K}_{10}{\rm M}_{00}$. The latter is equivalent to the system
\begin{subequations}\label{vertex-bond-DNLS}
    \begin{align}
        f_{00} g_{10} - g_{00} f_{01}&=0,\label{vertex-bond-DNLS-a}\\
      f_{01} q_{11} - f_{00} q_{10} -  g_{10} q_{11} +  g_{00} q_{01}&=0,\label{vertex-bond-DNLS-b}\\
 f_{01} p_{01} - f_{00} p_{00} -  g_{10} p_{10} +  g_{00} p_{00}&=0,\label{vertex-bond-DNLS-c}\\
f_{01} - f_{00} -  (g_{10} - g_{00}) - f_{00} g_{10} p_{10} q_{10} + g_{00} f_{01} p_{01} q_{01}&=0.\label{vertex-bond-DNLS-d}
    \end{align}
\end{subequations}
System \eqref{vertex-bond-DNLS} admits the first integrals
\begin{equation}\label{VB-DNLS-1-ints}
    \mathcal{F}=f_{00}^2 p_{00}q_{10}-f_{00},\qquad \mathcal{G}=g_{00}^2 p_{00}q_{01}-g_{00}.
\end{equation}

For the continuum limit, we consider the slow variables as in \eqref{slow-vars} with their shifts defined as in \eqref{xij-tij}. We have the following.

\begin{proposition}\label{DNLS-lim}
    The continuum limit of the discrete derivative NLS system \eqref{vertex-bond-DNLS} is the Kaup-Newell system \eqref{DNLS-system}.
\end{proposition}
\begin{proof}
    Let $p_{ij}=\sqrt{2h}p(x_{ij},t_{ij})$ and $q_{ij}=\sqrt{2h}q(x_{ij},t_{ij})$. Let also the auxiliary functions $f_{ij}$ and $g_{ij}$ be defined by the polynomial expressions:
    \begin{equation}\label{aux-fij-gij}
        f_{ij}=f_0(x_{ij},t_{ij})+f_1(x_{ij},t_{ij})h+f_2(x_{ij},t_{ij})h^2,\qquad g_{ij}=g_0(x_{ij},t_{ij})+g_1(x_{ij},t_{ij})h+g_2(x_{ij},t_{ij})h^2.
    \end{equation}
    
    Substituting to equation \eqref{vertex-bond-DNLS-a}, we obtain:
    $$
    (f_0g_{0,x}-g_0f_{0,x})h+\mathcal{O}(h^2)=0,
    $$
    from where it follows that $\partial_x \left(\frac{f_0}{g_0}\right)=0$, or $f_0=c(t)g_0$. Without loss of generality, we choose $f_0=g_0=1$. Now, from equation \eqref{vertex-bond-DNLS-c}, we obtain
    $$
    \left(p_{00}(f_{1,x}-g_{1,x})+(2+f_1-g_1)p_x\right)h^{\frac{5}{2}}+\mathcal{O}(h^{\frac72})=0.
    $$
    A simple choice that annihilates the above coefficient would be $f_1=-1$ and $g_1=1$. However, this choice will not generate the target system, and we need to take into account the first integrals \eqref{VB-DNLS-1-ints}. 

    Let us take the first integral $\mathcal{F}=f_{00}^2 p_{00}q_{10}-f_{00}$ from \eqref{VB-DNLS-1-ints}. Upon expansion, it can be readily seen that:
    $$
    \mathcal{F}=-1+h (2pq-f_1)+\mathcal{O}(h^2),
    $$
    where $f_1$ is the coefficient of $h$ in \eqref{aux-fij-gij}. This suggests choosing, for instance, $f_1=-1+2pq$, so that the $\mathcal{O}(h)$ is finite. Therefore, we adopt the ansatz $f_{00}=1-h+p_{00}q_{10}+h^2r_{00}+\mathcal{O}(h^3)$. Similarly, taking into account the second integral $\mathcal{G}=g_{00}^2 p_{00}q_{01}-g_{00}$ from \eqref{VB-DNLS-1-ints}, we adopt the ansatz that $g_{00}=1+h+p_{00}q_{01}+h^2s_{00}+\mathcal{O}(h^3)$. 

    That is, we substitute the following
    \begin{align*} 
    f_{00} &= 1 - h + p_{00} q_{10} + h^2 r_{00} + O(h^3),\\ 
    f_{01} &= 1 - h + p_{01} q_{11} + h^2 r_{01} + O(h^3),\\ 
    g_{00} &= 1 + h + p_{00} q_{01} + h^2 s_{00} + O(h^3),\\ 
    g_{10} &= 1 + h + p_{10} q_{11} + h^2 s_{10} + O(h^3),
    \end{align*}
    into system \eqref{vertex-bond-DNLS}. Equation \eqref{vertex-bond-DNLS-a} gives
    $$
    (r_x-s_x+8(pq)_x)h^3+\mathcal{O}(h^4)=0
    $$
    That is $r_x-s_x+8(pq)_x=0$. We choose the simple solution $r-s=-8pq$. In view of the latter, equations  \eqref{vertex-bond-DNLS-a} and  \eqref{vertex-bond-DNLS-d} are identically satisfied up to $\mathcal{O}(h^4)$, whereas equations \eqref{vertex-bond-DNLS-b} and \eqref{vertex-bond-DNLS-c} are equivalent to
    $$
    \sqrt{2} \left[2 p_t - p_{xx}+4\,(p^2q)_x\right]h^{\frac72}+\mathcal{O}(h^{\frac92}), \quad\text{and}\quad \sqrt{2} \left[2q_t +q_{xx}+4\,(q^2p)_x\right]h^{\frac72}+\mathcal{O}(h^{\frac92}),
    $$
    and setting the coefficients of $h^{\frac72}$ equal to zero, we obtain exactly system \eqref{DNLS-system}. Then the terms of order $\mathcal{O}(h^4)$ in  \eqref{vertex-bond-DNLS-a} and  \eqref{vertex-bond-DNLS-d} vanish modulo  $r-s=-8pq$ and the system \eqref{DNLS-system}.
\end{proof}

\subsection{Mikhailov--Shabat--Yamilov system} 
Finally, the Mikhailov--Shabat--Yamilov deformation of the DNLS system \cite{MSY} is given by
\begin{equation}\label{DDNLS-sys} 
p_t =
-\frac12\, p_{xx}
-2 (p^2 q)_x
+ q_x,
\quad q_t =
\frac12\, q_{xx}
-2 (q^2 p)_x
+ p_x.
\end{equation}
where $x\in\mathbb{R}$ and $t>0$, and the system is equivalent to equation \eqref{zero-curvature}, where the matrices ${\rm U}$ and ${\rm V}$ are given as follows: 
\begin{subequations}\label{Lax-DDNLS} 
\begin{align} 
U&=
-\lambda^2
\begin{pmatrix}
1 & 0\\
0 & -1
\end{pmatrix}
-\lambda
\begin{pmatrix}
0 & 2p\\
2q & 0
\end{pmatrix}
-\lambda^{-1}
\begin{pmatrix}
0 & 2q\\
2p & 0
\end{pmatrix}
+\lambda^{-2}
\begin{pmatrix}
1 & 0\\
0 & -1
\end{pmatrix}.
\\[2mm]
V&=-\lambda^4
\begin{pmatrix}
1 & 0\\
0 & -1
\end{pmatrix}
-\lambda^3
\begin{pmatrix}
0 & 2p\\
2q & 0
\end{pmatrix}
+\lambda^2
\begin{pmatrix}
2pq & 0\\
0 & -2pq
\end{pmatrix}
+\lambda
\begin{pmatrix}
0 & p_x+4p^2q-q\\
4pq^2-p-q_x & 0
\end{pmatrix}
+\nonumber\\
&\lambda^{-1}
\begin{pmatrix}
0 & 4pq^2-p-q_x\\
p_x+4p^2q-q & 0
\end{pmatrix}
-\lambda^{-2}
\begin{pmatrix}
2pq & 0\\
0 & -2pq
\end{pmatrix}
-\lambda^{-3}
\begin{pmatrix}
0 & 2q\\
2p & 0
\end{pmatrix}
+\lambda^{-4}
\begin{pmatrix}
1 & 0\\
0 & -1
\end{pmatrix}.
\end{align}  
\end{subequations}
In this case, matrices ${\rm U}$ and ${\rm V}$ are invariant under the transformation
\begin{subequations}\label{D2-symmetry}
    \begin{align}
    &\mathfrak{s}_1(\lambda):~{\rm U}(\lambda)\rightarrow \sigma_3{\rm U}(-\lambda)\sigma_3,~~\quad \mathfrak{s}_1(\lambda):~{\rm V}(\lambda)\rightarrow \sigma_3{\rm V}(-\lambda)\sigma_3,\\
    &\mathfrak{s}_2(\lambda):~{\rm U}(\lambda)\rightarrow \sigma_1{\rm U}\left(\frac{1}{\lambda}\right)\sigma_1,\quad \mathfrak{s}_2(\lambda):~{\rm V}(\lambda)\rightarrow \sigma_1{\rm V}\left(\frac{1}{\lambda}\right)\sigma_1.
\end{align}
\end{subequations}
The above involutions generate the reduction group, which is isomorphic to the dihedral group $\mathbb{Z}_2\times\mathbb{Z}_2\cong \mathbb{D}_2$.

In \cite{SPS} we constructed a Darboux transformation for system \eqref{DDNLS-sys} that consists of the Darboux matrix
\begin{equation}\label{Darboux-DDNLS}
     {\rm M}=\lambda^2
\begin{pmatrix} f&0\\0&0\end{pmatrix}
+\lambda
\begin{pmatrix}0&  f  p\\ \tilde{q}  f&0\end{pmatrix}
+\begin{pmatrix}  f  g&0\\0&  f  g\end{pmatrix}+\lambda^{-1}
\begin{pmatrix}0&f\tilde{q}  \\   f  p&0\end{pmatrix}
+\lambda^{-2}
\begin{pmatrix}0&0\\0&  f\end{pmatrix},
\end{equation}
and the $x$-part of the B\"acklund transformation
\begin{align}\label{x-Backlund-DDNLS}
    p_x&=2\big((\tilde{p}\tilde{q}-pq)p+(p-\tilde{p})g+q-\tilde{q}\big),\\
    \tilde{q}_x&=2\big((\tilde{p}\tilde{q}-pq)\tilde{q}+(q-\tilde{q})g+p-\tilde{p}\big),\\
    g_x&=2\big((\tilde{p}\tilde{q}-pq)g+(p-\tilde{p})p+(q-\tilde{q})\tilde{q}\big),\\
    f_x&=2(pq-\tilde{p}\tilde{q})f
\end{align}

For the discretisation of system \eqref{DDNLS-sys}, we consider two Darboux matrices \eqref{Darboux-DDNLS} in the following notation:
\begin{align*}
    {\rm M}_{00}&={\rm M}(p_{00},q_{10},f_{00},g_{00})=f_{00}\begin{pmatrix}\lambda^2+g_{00} & \lambda p_{00}+\lambda^{-1}q_{10}\\ \lambda q_{10}+\lambda^{-1}p_{00} &\lambda^{-2}+g_{00} \end{pmatrix},\\ 
    {\rm K}_{00}&={\rm M}(p_{00},q_{01},u_{00},v_{00})=u_{00}\begin{pmatrix}\lambda^2+v_{00} &\lambda p_{00}+\lambda^{-1}q_{01}\\ \lambda q_{01}+\lambda^{-1}p_{00} &\lambda^{-2}+v_{00} \end{pmatrix},
\end{align*}
and substitute them into the Lax equation ${\rm M}_{01}{\rm K}_{00}={\rm K}_{10}{\rm M}_{00}$ which is equivalent to the system of P$\Delta$Es:
\begin{subequations}\label{d2-discrete}
    \begin{align}
        f_{01}u_{00}-u_{10} f_{00}&=0,\\
        p_{00}(p_{01}-p_{00})+(q_{01}-q_{10})q_{11}+g_{01}v_{00}-g_{00}v_{10}&=0,\\
        g_{01}-g_{00}+v_{00}-v_{10}+p_{01}q_{01}-p_{10}q_{10}&=0,\\
        p_{01}-p_{10}+g_{01}q_{01}+(v_{00}-g_{00})q_{11}-v_{10}q_{10}&=0,\\
        q_{01}-q_{10}+g_{01}p_{00}-g_{00}p_{10}+v_{00}p_{01}-v_{10}p_{00}&=0.
    \end{align}
\end{subequations}

For the continuum limit, we consider the slow variables as in \eqref{slow-vars} with their shifts defined as in \eqref{xij-tij}. We have the following.

\begin{proposition}\label{DDNLS-lim}
    The continuum limit of the discrete Mikhailov--Shabat--Yamilov deformation of the derivative NLS system \eqref{d2-discrete} is system \eqref{DDNLS-sys}.
\end{proposition}
\begin{proof}
    Let
    \[
        p_{ij}=\sqrt{2}\,p(x_{ij},t_{ij}),\qquad
        q_{ij}=\sqrt{2}\,q(x_{ij},t_{ij}).
    \]
    In contrast to the Kaup--Newell case, no factor of $\sqrt{h}$ is introduced in the dependent variables. We make the following singular ansatz for the auxiliary functions:
    \begin{align*}
        g_{ij}&=-\frac{1}{h}-p_{ij}q_{i+1,j}+h r_{ij}+\mathcal{O}(h^2),\\
        v_{ij}&=\frac{1}{h}-p_{ij}q_{i,j+1}+h s_{ij}+\mathcal{O}(h^2),
    \end{align*}
    where $r_{ij}=r(x_{ij},t_{ij})$ and $s_{ij}=s(x_{ij},t_{ij})$. That is,
    \begin{align*}
        g_{00}&=-\frac{1}{h}-p_{00}q_{10}+h r_{00}+\mathcal{O}(h^2),&
        g_{01}&=-\frac{1}{h}-p_{01}q_{11}+h r_{01}+\mathcal{O}(h^2),\\
        v_{00}&=\frac{1}{h}-p_{00}q_{01}+h s_{00}+\mathcal{O}(h^2),&
        v_{10}&=\frac{1}{h}-p_{10}q_{11}+h s_{10}+\mathcal{O}(h^2).
    \end{align*}
    
    Substituting the above expressions into equation \eqref{d2-discrete} yields, from the third equation, $-(r_x+s_x)h^2+\mathcal{O}(h^3)=0$.
    Thus, $ r_x+s_x=0$. The second equation of \eqref{d2-discrete}, after using the above relation together with the resulting evolution equations, gives $r_x-s_x-\left(8p^2q^2-4p^2-4q^2\right)_x=0$. Therefore, we may choose
    $$
    r=4p^2q^2-2p^2-2q^2,\qquad s=-4p^2q^2+2p^2+2q^2.
    $$
    
    With this choice, the fourth and fifth equations of \eqref{d2-discrete} give, respectively,
    \[
        \sqrt{2}\left[
        2q_t-q_{xx}+4(q^2p)_x-2p_x
        \right]h+\mathcal{O}(h^2)=0,
    \]
    and
    \[
        \sqrt{2}\left[
        2p_t+p_{xx}+4(p^2q)_x-2q_x
        \right]h+\mathcal{O}(h^2)=0.
    \]
    Setting the coefficients of $h$ equal to zero, we obtain
    \[
        2p_t+p_{xx}+4(p^2q)_x-2q_x=0,
        \qquad
        2q_t-q_{xx}+4(q^2p)_x-2p_x=0.
    \]
    Equivalently,
    \[
        p_t=-\frac12p_{xx}-2(p^2q)_x+q_x,
        \qquad
        q_t=\frac12q_{xx}-2(q^2p)_x+p_x,
    \]
    which is precisely system \eqref{DDNLS-sys}. Hence the continuum limit of \eqref{d2-discrete} is the Mikhailov--Shabat--Yamilov deformation of the derivative NLS system.
\end{proof}

\section{Noncommutative NLS-type equations and solutions}\label{NC-NLS-eqs}
In this section, we recall some noncommutative versions of the Ablowitz--Kaup--Newell--Segur (AKNS) NLS system and the Kaup--Newell DNLS system \cite{Kaup-Newell} together with their Darboux--B\"acklund transformations \cite{KRX} and the Mikhailov--Shabat--Yamilov deformation of the DNLS system \cite{MSY, Sokor-Nikitina}. Moreover, we present a noncommutative version of the deformation of the DNLS system with its Lax pair. We construct a Darboux matrix and a B\"acklund transformation for it.

\subsection{Noncommutative NLS-type systems}\label{NCsystemsNLS}
Let $\mathcal{R}$ be a noncommutative division ring, and let $Z(\mathcal{R})$ denote its centre. Assume also that the functions $\bm{p}=\bm{p}(x,t)$ and $\bm{q}=\bm{q}(x,t)$ are elements of $\mathcal{R}$. 

The construction of noncommutative analogues of the AKNS and the DNLS systems is relatively straightforward, and it is almost done by replacing $p$ and $q$ in their Lax pairs \eqref{Lax-NLS} and \eqref{Lax-DDNLS}, respectively. However, the construction of the noncommutative version of system \eqref{DDNLS-sys} is more subtle, and we shall construct it in this section.

In particular, recall the noncommutative version of system \eqref{NLS-eq}, namely the system:
\begin{equation}\label{NLS-eq-NC}
\bm{p}_t=\frac{1}{2}\bm{p}_{xx}-4\bm{p}\bm{q}\bm{p}, ~~~ \bm{q}_t=-\frac{1}{2}\bm{q}_{xx}+4\bm{q}\bm{p}\bm{q},
\end{equation}
which was derived in \cite{OS} within the classification of integrable evolutionary equations on associative algebras. The first matrix generalization of the NLS equation is due to Manakov \cite{Manakov}. System \eqref{NLS-eq-NC} is equivalent to the zero-curvature condition \eqref{zero-curvature}, where
\begin{equation}\label{Lax-NLS-NC}
    {\rm{U}}=\lambda\begin{pmatrix}
   1 & 0\\
   0 & -1
\end{pmatrix}
+
\begin{pmatrix}
   0 & 2\bm{p}\\
   2\bm{q} & 0
\end{pmatrix}, \quad
 {\rm{V}}=\lambda^2 \begin{pmatrix}
   1 & 0\\
   0 & -1
\end{pmatrix}
+\lambda \begin{pmatrix}
   0 & 2\bm{p}\\
   2\bm{q} & 0
\end{pmatrix}
+
\begin{pmatrix}
   -2\bm{p}\bm{q} & \bm{p}_x\\
   -\bm{q}_x & 2\bm{q}\bm{p}
\end{pmatrix}.
\end{equation}  

Moreover, recall the noncommutative derivative NLS equation  \cite{OS} reads
\begin{equation}\label{DNLS-system-NC}
        \bm{p}_t = \frac{1}{2} \bm{p}_{xx}-2\,(\bm{p} \bm{q} \bm{p})_x,\quad \bm{q}_t = -\frac{1}{2}\bm{q}_{xx}-2\,(\bm{q} \bm{p} \bm{q})_x,
\end{equation}
and it is equivalent to equation \eqref{zero-curvature}, where the matrices ${\rm U}$ and ${\rm V}$ are given as follows:
\begin{subequations}\label{Lax-NC-DNLS}
    \begin{align}
          {\rm U}&=\lambda^2\begin{pmatrix}
1 & 0\\
0 & -1
\end{pmatrix}
+
\lambda
\begin{pmatrix}
0 & 2\bm{p}\\
2\bm{q} & 0
\end{pmatrix},
\label{Lax-NC-DNLS-a}\\
{\rm V}&=
\lambda^4
\begin{pmatrix}
1 & 0\\
0 & -1
\end{pmatrix}
+
\lambda^3
\begin{pmatrix}
0 & 2\bm{p}\\
2\bm{q} & 0
\end{pmatrix}
-2\lambda^2
\begin{pmatrix}
\bm{p}\bm{q} & 0\\
0 & -\bm{q}\bm{p}
\end{pmatrix}
+
\lambda
\begin{pmatrix}
0 & \bm{p}_x-4\bm{p}\bm{q}\bm{p}\\
-\bm{q}_x-4\bm{q}\bm{p}\bm{q} & 0
\end{pmatrix}.\label{Lax-NC-DNLS-b}
    \end{align}
\end{subequations}

As can be seen above, the noncommutative NLS and DNLS systems \eqref{NLS-eq-NC} and \eqref{DNLS-system-NC} are obtained by replacing in the Lax pairs \eqref{Lax-NLS} and \eqref{Lax-DNLS}, of their commutative analogues, $p$ and $q$ by ${\bm p}$ and ${\bm q}$, respectively, also replacing the terms $pq$ that appear on the diagonals of matrices ${\rm V}$ in \eqref{Lax-NLS} and \eqref{Lax-DNLS-b} by ${\bm p}{\bm q}$ and ${\bm q}{\bm p}$ on the opposite sides along the diagonals, and finally replacing terms $p^2q$ and $q^2p$ in \eqref{Lax-DNLS-b} by ${\bm p}{\bm q}{\bm p}$ and ${\bm q}{\bm p}{\bm q}$, respectively. However, arbitrary noncommutative replacements in the Lax pair \eqref{Lax-DDNLS} of the deformation of the DNLS system \eqref{DDNLS-sys} impose unwanted commutativity relations between ${\bm p}$ and ${\bm q}$ after substitution of the associated `noncommutative' Lax pair into the zero-curvature condition. 

In order to construct a noncommutative version of system \eqref{DDNLS-sys} without any commutativity relations between ${\bm p}$ and ${\bm q}$, it becomes necessary to introduce nonlocal variables in its $t$-part of the Lax pair, namely matrix ${\bm V}$. In particular, we have the following.

\begin{theorem}(Noncommutative deformation of the DNLS system)\label{NC-DDNLS-th}
    Let the functions $\bm{p}=\bm{p}(x,t)$ and $\bm{q}=\bm{q}(x,t)$ be elements of a noncommutative division ring $\mathcal{R}$. Then, the following noncommutative deformation of the DNLS system   
\begin{subequations}\label{system-def-DNLS}  
    \begin{align}
     \bm p_t &=
-\frac12 \bm p_{xx}
-(\bm p\bm S)_x
+\bm q_x,
\\
\bm q_t &=
\frac12 \bm q_{xx}
-(\bm q\bm S)_x
+\bm p_x, 
\end{align}
\end{subequations}
where $\bm S=\bm q\bm p+\bm p\bm q$, has a Lax pair given by
\begin{equation}\label{Lax-DDNLS-NC} 
{\rm U}=-\lambda^2
\begin{pmatrix}
1 & 0\\
0 & -1
\end{pmatrix}
-\lambda
\begin{pmatrix}
0 & 2{\bm p}\\
2{\bm q} & 0
\end{pmatrix}
-\lambda^{-1}
\begin{pmatrix}
0 & 2{\bm q}\\
2{\bm p} & 0
\end{pmatrix}
+\lambda^{-2}
\begin{pmatrix}
1 & 0\\
0 & -1
\end{pmatrix},\qquad
{\rm V}=\begin{pmatrix}
{\rm V}_{11} & {\rm V}_{12}\\
{\rm V}_{21} & {\rm V}_{22}
\end{pmatrix},
\end{equation}
where the off–diagonal entries of {\rm V} are given by
\begin{equation}
\label{def-DNLS-V-offdiag-corr}
{\rm V}_{12}=-2\lambda^3\bm p-2\lambda^{-3}\bm q+\lambda \bm A+\lambda^{-1}\bm B,
\qquad
{\rm V}_{21}=-2\lambda^3\bm q-2\lambda^{-3}\bm p+\lambda \bm B+\lambda^{-1}\bm A,
\end{equation}
where $\bm A=\bm p_x+2\bm p\bm S$ and $\bm B=-\bm q_x+2\bm q\bm S$, and the diagonal entries of {\rm V} are given by
\begin{equation}
\label{def-DNLS-V-diag-corr}
{\rm V}_{11}=-\lambda^4+\lambda^2\alpha+\beta+\lambda^{-2}\gamma+\lambda^{-4},\qquad {\rm V}_{22}=\lambda^4+\lambda^2\gamma+\beta+\lambda^{-2}\alpha-\lambda^{-4},
\end{equation}
where  $\alpha,\beta,\gamma\in\mathfrak{R}$ satisfy the algebraic relations
\begin{subequations}\label{eq-alpha-gamma-p}
\begin{align}
\alpha\bm p-\bm p\gamma=2(\bm p\bm S+\bm q),&\quad \gamma\bm q-\bm q\alpha=-2(\bm q\bm S+\bm p),\label{eq-alpha-gamma-p-a}\\
[\beta,\bm p]=\bm q\gamma-\alpha\bm q+2(\bm q\bm S+\bm p),&\quad [\beta,\bm q]=\bm p\alpha-\gamma\bm p-2(\bm p\bm S+\bm q),\label{eq-alpha-gamma-p-b}
\end{align}
\end{subequations}
and the differential equations
\begin{subequations}\label{eq-alpha-beta-gamma-x}
\begin{align}
\alpha_x&=-4\bm p^2+4\bm q^2-2\bm p\bm B+2\bm A\bm q,\label{eq-alpha-x}\\
\beta_x&=2\bm A\bm p-2\bm p\bm A+2\bm B\bm q-2\bm q\bm B,\label{eq-beta-x}\\
\gamma_x&=4\bm p^2-4\bm q^2-2\bm q\bm A+2\bm B\bm p .\label{eq-gamma-x}
\end{align}
\end{subequations}
\end{theorem}

\begin{proof}
We need to prove that the zero-curvature condition
$$
\mathcal{F}:=\rm{U}_t-\rm{V}_x+[\rm{U},\rm{V}]=0,
$$
in view of relations \eqref{eq-alpha-gamma-p} and \eqref{eq-alpha-beta-gamma-x}.

    Let $\rm{V}^{11}$ and $\rm{V}^{22}$ be solutions of the following differential equations
    \begin{equation}\label{Vii-ODEs}
        \rm{V}^{11}_x=\rm{U}^{12}\rm{V}^{21}-\rm{V}^{12}\rm{U}^{21},\quad \rm{V}^{22}_x=\rm{U}^{21}\rm{V}^{11}-\rm{V}^{21}\rm{U}^{12}.
    \end{equation}
Then, the diagonal equations, $\mathcal{F}_{ii}$, $i=1,2$, of the zero-curvature condition read:
\begin{align*}
    \mathcal{F}_{11}&=\rm{U}^{11}_t-\rm{V}^{11}_x+\rm{U}^{11}\rm{V}^{11}+\rm{U}^{12}\rm{V}^{21}-\rm{V}^{11}\rm{U}^{11}-\rm{V}^{12}\rm{U}^{21},\\
    \mathcal{F}_{22}&=\rm{U}^{22}_t-\rm{V}^{22}_x+\rm{U}^{21}\rm{V}^{12}+\rm{U}^{22}\rm{V}^{22}-\rm{V}^{21}\rm{U}^{12}-\rm{V}^{22}\rm{U}^{22},
\end{align*}
which are identically satisfied in view of $\rm{U}^{11}_t=\rm{U}^{22}_t=0$ and the definitions \eqref{Vii-ODEs} of $\rm{V}^{ii}$, $i=1,2$.

Next, let $\mathcal{K}= \rm{U}^{12}_t-\rm{V}^{12}_x+2\rm{U}^{11}\rm{V}^{12}$ and $\Theta=\rm{U}^{12}\rm{V}^{22}-\rm{V}^{11}\rm{U}^{12}$. Then, the equation coming from the $\mathcal{F}_{12}$ element is given by:
\begin{align*}
    \mathcal{F}_{12}=&\rm{U}^{12}_t-\rm{V}^{12}_x+\rm{U}^{11}\rm{V}^{12}+\rm{U}^{12}\rm{V}^{22}-\rm{V}^{11}\rm{U}^{12}-\rm{V}^{12}\rm{U}^{22}\\
    =&\rm{U}^{12}_t-\rm{V}^{12}_x+2\rm{U}^{11}\rm{V}^{12}+\rm{U}^{12}\rm{V}^{22}-\rm{V}^{11}\rm{U}^{12}\\
    =&\mathcal{K}+\Theta.
\end{align*}
The term $\mathcal{K}$ is written in details as:
\begin{equation}\label{kappa}
    \mathcal{K}=4{\bm p} \lambda^5+(2{\bm p}_x-2A)\lambda^3-(2{\bm p}_t+A_x+2B+4{\bm p})\lambda-(2{\bm q}_t+B_x-2A-4{\bm q})\lambda^{-1}+(2{\bm q}_x+2B)\lambda^{-3}-4{\bm q}\lambda^{-5}.
\end{equation}

By definition \eqref{Vii-ODEs} of $\rm{V}^{ii}$, $i=1,2$, it follows that the latter are Laurent polynomials with coefficients $\left\{\lambda^{\pm 4},\lambda^{\pm 2}\right\}$. Thus,  $\rm{V}^{ii}$, $i=1,2$, are of the following form
\begin{subequations}\label{Vii}
    \begin{align}
    {\rm V}^{11}&=a_4\lambda^4+a_2\lambda^2+a_0+a_{-2}\lambda^{-2}+a_{-4}\lambda^{-4},\label{V11}\\
    {\rm V}^{22}&=d_4\lambda^4+d_2\lambda^2+d_0+d_{-2}\lambda^{-2}+d_{-4}\lambda^{-4}.\label{V22}
\end{align}
\end{subequations}

In this notation of $\rm{V}^{ii}$, $i=1,2$, $\Theta$ can be written as:
\begin{align}
    \Theta=&2(a_4{\bm p}-{\bm p}d_4)\lambda^5+2(a_2{\bm p}+a_4{\bm q}-{\bm p}d_2-{\bm q}d_4)\lambda^3+2(a_0{\bm p}+a_2{\bm q}-{\bm p}d_0-{\bm q}d_2)\lambda+\nonumber\\&2(a_0{\bm q}+a_{-2}{\bm p}-{\bm q}d_0-{\bm p}d_{-2})\lambda^{-1}+2(a_{-2}{\bm q}+a_{-4}{\bm p}-{\bm q}d_{-2}-{\bm p}d_{-4})\lambda^{-3}+2(a_{-4}{\bm q}-{\bm q}d_{-4})\lambda^{-5}.\label{theta-laurent}
\end{align}
Moreover, from \eqref{Vii-ODEs} and \eqref{Vii} we obtain:
\begin{subequations}\label{Vii-a-d}
    \begin{align}
\rm{V}^{11}_x&=a_{4,x}\lambda^4+a_{2,x}\lambda^2+a_{0,x}+a_{-2,x}\lambda^{-2}+a_{-4,x}\lambda^{-4}=\rm{U}^{12}\rm{V}^{21}-\rm{V}^{12}\rm{V}^{21},\label{Vii-a-d-a}\\
\rm{V}^{22}_x&=d_{4,x}\lambda^4+d_{2,x}\lambda^2+d_{0,x}+d_{-2,x}\lambda^{-2}+d_{-4,x}\lambda^{-4}=\rm{U}^{21}\rm{V}^{11}-\rm{V}^{21}\rm{V}^{12}.\label{Vii-a-d-b}
\end{align}
\end{subequations}

Equation \eqref{Vii-a-d-a} is equivalent to the system
\begin{align*}
    \lambda^4:&\quad a_{4,x}=0,\\
    \lambda^2:&\quad a_{2,x}=-4\bm{p}^2+4\bm{q}^2+2A\bm{q}-2\bm{p}B,\\
    \lambda^0:&\quad a_{0,x}=-2\bm{p}A-2\bm{q}B+2A\bm{p}+2B\bm{q},\\
    \lambda^{-2}:&\quad a_{-2,x}=4\bm{p}^2-4\bm{q}^2-2\bm{q}A+2B\bm{p},\\
    \lambda^{-4}:&\quad a_{-4,x}=0,
\end{align*}
whereas equation \eqref{Vii-a-d-b} is equivalent to the system
\begin{align*}
    \lambda^4:&\quad d_{4,x}=0,\\
    \lambda^{2}:&\quad d_{2,x}=4\bm{p}^2-4\bm{q}^2-2\bm{q}A+2B\bm{p},\\
    \lambda^0:&\quad d_{0,x}=-2\bm{p}A-2\bm{q}B+2A\bm{p}+2B\bm{q},\\
    \lambda^{-2}:&\quad d_{-2,x}=-4\bm{p}^2+4\bm{q}^2+2A\bm{q}-2\bm{p}B\\
    \lambda^{-4}:&\quad d_{-4,x}=0.
\end{align*}
That is, $a_{\pm 4}$ and $d_{\pm 4}$ are constants; we choose $a_4=d_{-4}=-1$ and $\alpha_{-4}=d_4=1$. Moreover, it follows from the above that $a_{0,x}=d_{0,x}$, $a_{2,x}=d_{-2,x}$, and $a_{-2,x}=d_{2,x}$; we set $a_2=d_{-2}=\alpha$, $a_0=d_0=\beta$, and $a_{-2}=d_{2}=\gamma$. Therefore, $\Theta$ from \eqref{theta-laurent} can be now written as:
\begin{align}
    \Theta=&-4{\bm p}\lambda^5+2(\alpha{\bm p}-2{\bm q}-{\bm p}\gamma)\lambda^3+2(\beta{\bm p}+\alpha{\bm q}-{\bm p}\beta-{\bm q}\gamma)\lambda+\nonumber\\&2(\beta{\bm q}+\gamma{\bm p}-{\bm q}\beta-{\bm p}\alpha)\lambda^{-1}+2(\gamma{\bm q}+2{\bm p}-{\bm q}\alpha)\lambda^{-3}+4{\bm q}\lambda^{-5}.\label{theta-laurent-new}
\end{align}

Therefore, from \eqref{kappa} and \eqref{theta-laurent-new} follows that:
\begin{align}
    \mathcal{F}_{12}=&2 (\alpha{\bm p}-{\bm p}\gamma-2{\bm p}S-2{\bm q})\lambda^3-(2{\bm p}_t+A_x+2B+4{\bm p}+2 \beta{\bm p}+2\alpha{\bm q}-2{\bm p}\beta-2{\bm q}\gamma)\lambda-\nonumber\\
    &(2{\bm q}_t+B_x-2A-4{\bm q}+2 \beta{\bm q}+2\gamma{\bm p}-2{\bm q}\beta-2{\bm p}\alpha)\lambda^{-1}+2(\gamma{\bm q}-{\bm q}\alpha+2{\bm q}S+2{\bm p})\lambda^{-3}=0,
\end{align}
for any $\lambda\in\mathbb{C}$. The coefficients of $\lambda^{\pm 3}$ are identically zero due to \eqref{eq-alpha-gamma-p-a}. Then, the coefficients $\lambda^{\pm 1}$ are identically zero if and only if system \eqref{system-def-DNLS} holds together with the restrictions \eqref{eq-alpha-gamma-p-b} on the Laurent coefficients of ${\rm V}_{ii}$, $i=1,2$. Exactly the same equations we obtain from $\mathcal{F}_{21}=0$, due to symmetry. 

Finally, equations \eqref{Vii-ODEs} can be written equivalently as \eqref{eq-alpha-beta-gamma-x}.
\end{proof}

\subsection{Solutions of the noncommutative NLS-type systems}
Here, we consider the noncommutative ring of matrices $n\times n$. The (matrix) noncommutative systems \eqref{NLS-eq-NC}, \eqref{DNLS-system-NC} and \eqref{system-def-DNLS} admit nontrivial solutions. The simplest way to demonstrate this is to `lift' solutions to the (scalar) commutative analogues of these systems to matrix solutions that do not commute with each other.

We demonstrate how to lift scalar solutions to matrix solutions. Start with the two-parametric solution
$$
p(x,t)=\frac{e^{x+\frac12t+\theta_0}}
    {1+e^{2x+\theta_0+\eta_0}},
    \qquad
    q(x,t)=-\frac{e^{x-\frac12t+\eta_0}}
    {1+e^{2x+\theta_0+\eta_0}},
$$
of the AKNS system \eqref{NLS-eq}. This can be lifted to the matrix solution
$$
    \bm p(x,t)=\frac{e^{x+\frac12t+\theta_0}}
    {1+e^{2x+\theta_0+\eta_0}}{\rm A}_n,
    \quad
    \bm q(x,t)=-\frac{e^{x-\frac12t+\eta_0}}
    {1+e^{2x+\theta_0+\eta_0}}{\rm B}_n,\qquad {\rm A}_n=\sum_{i=1}^{n-1}E_{i,i+1},
    \quad
    {\rm B}_n=\sum_{i=1}^{n-1}E_{i+1,i},
$$
where $E_{ij}$ is the standard matrix unit, of the noncommutative AKNS system \eqref{NLS-eq-NC}. By construction, $\bm p(x,t)$ and $\bm q(x,t)$ do not commute, since
$$
{\rm A}_n{\rm B}_n=\diag(1,\dots,1,0), \qquad {\rm A}_n{\rm B}_n=\diag(0,1,\dots,1),
$$
so ${\rm A}_n{\rm B}_n\neq {\rm B}_n{\rm A}_n$, for $n\geq2$.

Next, it can be  verified that the travelling wave functions
$$
p(x,t)=ae^{kx+\omega t},\qquad q(x,t)=be^{-kx-\omega t},
$$
are solutions of the commutative Kaup--Newell system \eqref{DNLS-system} if $\omega=\frac{k^2}{2}-2abk$. This solution can be lifted to a matrix solution
\begin{equation}\label{eq:DNLS-matrix-sol}
    \bm p(x,t)=a e^{kx+(\frac{k^2}{2}-2abk)t}{\rm A}_n,
    \quad
    \bm q(x,t)=b e^{-kx-(\frac{k^2}{2}-2abk)t}{\rm B}_n, \qquad {\rm A}_n=\sum_{i=1}^{n-1}E_{i,i+1},
    \quad
    {\rm B}_n=\sum_{i=1}^{n-1}E_{i+1,i}.
\end{equation}

However, we can construct noncommutative solutions that are not necessarily lifts of scalar ones. In particular, for the noncommutative Mikhailov--Shabat--Yamilov system \eqref{system-def-DNLS}, we have the following.

\begin{proposition}
The noncommutative Mikhailov--Shabat--Yamilov system \eqref{system-def-DNLS} admits the travelling wave solution
    \begin{equation}\label{eq:def-matrix-sol}
    \bm p(x,t)=e^{kx+\omega t}{\rm A},
    \qquad
    \bm q(x,t)=\alpha e^{kx+\omega t} {\rm A}+h{\rm B}, \qquad  \alpha=\frac{\omega+\frac12k^2}{k},\quad \omega^2=k^2+\frac{k^4}{4},
    \quad k\neq0,
\end{equation}
where $h$ is a constant and ${\rm A},{\rm B}$ are constant $n\times n$ matrices satisfying
\begin{equation}\label{eq:nilpotent-algebra}
    {\rm A}^2={\rm B}^2=0,
    \qquad
    {\rm A}{\rm B}+{\rm B}{\rm A}=0,
    \qquad
    {\rm A}{\rm B}\neq0.
\end{equation}
\end{proposition}
\begin{proof}
    We seek a solution of the form
    $$
      \bm p(x,t)=f(x,t){\rm A}, \qquad 
      \bm q(x,t)=g(x,t){\rm A}+h{\rm B},
    $$
    where $h$ is constant and ${\rm A},{\rm B}$ satisfy \eqref{eq:nilpotent-algebra}. Then
    $$
    \bm S=\bm q\bm p+\bm p\bm q
    =gf{\rm A}^2+hf{\rm B}{\rm A}
    +fg{\rm A}^2+fh{\rm A}{\rm B}
    =fh({\rm B}{\rm A}+{\rm A}{\rm B})=0.
    $$
    Hence, system \eqref{system-def-DNLS} reduces to the linear system
    $$
    \bm p_t=-\frac12\bm p_{xx}+\bm q_x,
    \qquad
    \bm q_t=\frac12\bm q_{xx}+\bm p_x .
    $$
    Substituting the above ansatz gives
    $$
    f_t=-\frac12 f_{xx}+g_x,
    \qquad
    g_t=\frac12 g_{xx}+f_x .
    $$
    Now, we seek solutions of the form $f=e^{kx+\omega t}$ and $g=\alpha e^{kx+\omega t}$. The first of the above equations gives $\omega=-\frac12 k^2+\alpha k$, and therefore, since $k\neq0$, $\alpha=\frac{\omega+\frac12 k^2}{k}$. The  second equation gives $\alpha\omega=\frac12\alpha k^2+k$. Therefore
    $$
    \omega^2=k^2+\frac{k^4}{4},
    $$
    and
    $$
    \bm p(x,t)=e^{kx+\omega t}{\rm A},
    \qquad
    \bm q(x,t)=\alpha e^{kx+\omega t}{\rm A}+h{\rm B},
    $$
    with $\alpha$ and $\omega$ as in \eqref{eq:def-matrix-sol}, is a solution of
    \eqref{system-def-DNLS}.
\end{proof}

\section{Noncommutative Darboux--B\"acklund transformations of NLS-type}\label{NCDBT}
In this section, we recall the Darboux transformations, which were constructed in \cite{KRX, Sokor-Nikitina} for systems \eqref{NLS-eq-NC} and \eqref{DNLS-system-NC}. In the latter, we presented only the $x$-parts of the associated B\"acklund transformations, since the systems under consideration are evolutionary. This means that the systems themselves play the role of the $t$-part of the B\"acklund transformations. However, the Darboux transformations depend not only on the potential functions ${\bm p}$ and ${\bm q}$, but also on auxiliary functions. In this section, we also present the evolutions of these auxiliary functions for completeness.

Moreover, we derive a first integral for the B\"acklund transformation of system \eqref{DNLS-system-NC} and we use it to derive a noncommutative Volterra equation. Then, we will construct a Darboux matrix and a B\"acklund transformation for the noncommutative deformation of the DNLS system \eqref{system-def-DNLS}.

\subsection{NLS system (AKNS)}
In \cite{KRX} we found all the Darboux matrices for system \eqref{NLS-eq-NC} that are linear in the spectral parameter $\lambda$, i.e. ${\rm M}=\lambda {\rm M}^1+{\rm M}^0$, and have rank 1, i.e. $\rank({\rm M}^1)=1$. In particular, we proved the following theorem, where we have now derived the evolutions of the auxiliary functions.

\begin{theorem}
Let $\rm{M}=\lambda {\rm M}^1 +{\rm M}^0$, where $\rank {\rm M}^1=1$, be a Darboux matrix that leaves covariant the Lax pair \eqref{Lax-NLS-NC} of the noncommutative NLS system \eqref{NLS-eq-NC}. All the Darboux transformations associated with matrices of this form fall into one of the following two cases.
\begin{enumerate}
\item The Darboux transformation consists of the Darboux matrix
\begin{subequations} \label{eq:T-10}
\begin{equation}\label{DT-10}
{\rm{M}}(\bm{f},\bm{p},\tilde{\bm{q}})=\lambda\begin{pmatrix}
1 & 0 \\
0 & 0
\end{pmatrix}
+
\begin{pmatrix}
\bm{f} & \bm{p} \\
\tilde{\bm{q}} & 1
\end{pmatrix},
\end{equation}
and the associated $x$-part of the B\"acklund transformation:
\begin{equation}\label{BT-10}
\bm{f}_x=2\tilde{\bm{p}}\tilde{\bm{q}}-2\bm{p}\bm{q},\quad\bm{p}_x=2\tilde{\bm{p}} -2\bm{f}\bm{p},\quad \tilde{\bm{q}}_x=2\tilde{\bm{q}}\bm{f}-2\bm{q}.
\end{equation}
The $t$-part consists of system \eqref{NLS-eq-NC} together with the system
\begin{equation}\label{NLS-BT-t-part}
\tilde{\bm{p}}_t=\frac{1}{2}\tilde{\bm{p}}_{xx}-4\tilde{\bm{p}}\tilde{\bm{q}}\tilde{\bm{p}}, ~~~ \tilde{\bm{q}}_t=-\frac{1}{2}\tilde{\bm{q}}_{xx}+4\tilde{\bm{q}}\tilde{\bm{p}}\tilde{\bm{q}},
\end{equation}
and the evolution of the auxiliary function
\begin{equation}\label{f-evol-NLS}
    {\bm f}_t=2{\bm f}{\bm p}{\bm q} -2\tilde{{\bm p}}\tilde{{\bm q}}{\bm f}+{\bm p}{\bm q}_x+\tilde{{\bm p}}_x\tilde{{\bm q}}.
\end{equation}
\end{subequations}

\item The Darboux transformation consists of the Darboux matrix
\begin{subequations} \label{eq:T-01}
\begin{equation}\label{DT-01}
{\rm K}(\bm{g},\tilde{\bm{p}},\bm{q})=\lambda\begin{pmatrix}
0 & 0 \\
0 & -1
\end{pmatrix}
+
\begin{pmatrix}
1 & \tilde{\bm{p}} \\
 \bm{q} & \bm{g}
\end{pmatrix},
\end{equation}
and the associated $x$-part of the B\"acklund transformation:
\begin{equation}\label{BT-01}
\bm{g}_x=2\tilde{\bm{q}}\tilde{\bm{p}}-2\bm{q}\bm{p},\quad \tilde{\bm{p}}_{x}=2\tilde{\bm{p}}\bm{g}-2\bm{p},\quad \bm{q}_x=2\tilde{\bm{q}} -2\bm{g}\bm{q}.
\end{equation}
The $t$-part consists of system \eqref{NLS-eq-NC}, system \eqref{NLS-BT-t-part} and the evolution of the auxiliary function
\begin{equation}\label{g-evol-NLS}
    {\bm g}_t=2\tilde{{\bm p}}\tilde{{\bm q}}{\bm g}-2{\bm g}{\bm p}{\bm q} -{\bm q}{\bm p}_x-\tilde{{\bm q}}_x\tilde{{\bm p}}.
\end{equation}
\end{subequations}
\end{enumerate}
\end{theorem}

\subsection{Derivative NLS system}
In \cite{Sokor-Nikitina} we constructed all the possible Darboux transformations of the form ${\rm M}=\lambda^2 {\rm M}^2+\lambda {\rm M}^1+{\rm M}^0$, with $\rank({\rm M}^2)=1$, for system \eqref{DNLS-system-NC}. In particular we have the following.

\begin{theorem}
Let ${\rm M}=\lambda^2 {\rm M}^2+\lambda {\rm M}^1+{\rm M}^0$, where $\rank {\rm M}^2=1$, be a Darboux matrix that leaves covariant the Lax pair \eqref{Lax-NLS-NC} of the noncommutative derivative NLS system \eqref{DNLS-system-NC}. All the Darboux transformations associated with matrices of this form fall into one of the following two cases.
\begin{enumerate}
\item The Darboux transformation consists of the Darboux matrix
\begin{subequations}\label{DT-BT-DNLS-1}
\begin{equation}\label{DT-DNLS-1}
{\rm M} =\lambda^2\begin{pmatrix} \bm{f} & 0\\ 0 & 0\end{pmatrix}+\lambda \begin{pmatrix} 0 & \bm{f}\bm{p}\\ \tilde{\bm{q}}\bm{f} & 0 \end{pmatrix}+\begin{pmatrix} c_1 & 0\\ 0 & c_2\end{pmatrix}
\end{equation}
together with the associated $x$-part of the B\"acklund transformation:
\begin{align}\label{BT-DNLS-1}
\bm{f}_x=2(\tilde{\bm{p}}\tilde{\bm{q}}\bm{f}-\bm{f}\bm{p}\bm{q}),\quad
    (\bm{f}\bm{p})_x=2c_2\tilde{\bm{p}}-2c_1\bm{p},\quad
    (\tilde{\bm{q}}\bm{f})_x=2c_1\tilde{\bm{q}}-2c_2\bm{q}.
\end{align}
The $t$-part consists of system \eqref{DNLS-system-NC} together with the system
\begin{equation}\label{DNLS-BT-t-part}
\tilde{\bm{p}}_t = \frac{1}{2} \tilde{\bm{p}}_{xx}-2\,(\tilde{\bm{p}} \tilde{\bm{q}} \tilde{\bm{p}})_x,\quad \tilde{\bm{q}}_t = -\frac{1}{2}\tilde{\bm{q}}_{xx}-2\,(\tilde{\bm{q}} \tilde{\bm{p}} \tilde{\bm{q}})_x 
\end{equation}
and the evolution of the auxiliary function
\begin{equation}\label{f-evol-DNLS}
   {\bm f}_t={\bm f}{\bm p}({\bm q}_x+4{\bm q}{\bm p}{\bm q})+(\tilde{\bm{p}}_x-4\tilde{{\bm p}}\tilde{{\bm q}}\tilde{{\bm p}})\tilde{{\bm q}}{\bm f}+2c_1({\bm p} {\bm q}-\tilde{{\bm p}}\tilde{{\bm q}}).
\end{equation}
\end{subequations}

\item The Darboux transformation consists of the Darboux matrix
\begin{subequations}\label{DT-BT-DNLS-2}
\begin{equation}\label{DT-DNLS-2}
{\rm M} =\lambda^2\begin{pmatrix} 0 & 0\\ 0 & \bm{g}\end{pmatrix}+\lambda \begin{pmatrix} 0 & -\tilde{\bm{p}}\bm{g}\\ -\bm{g}\tilde{\bm{q}} & 0 \end{pmatrix}+\begin{pmatrix} c_1 & 0\\ 0 & c_2\end{pmatrix}
\end{equation}
together with the associated $x$-part of the B\"acklund transformation:
\begin{align}\label{BT-DNLS-2}
\bm{g}_x=2(\bm{g}\bm{q}\bm{p}-\tilde{\bm{q}}\tilde{\bm{p}}\bm{g}),\quad
    (\tilde{\bm{p}}\bm{g})_x=2c_1\bm{p}-2c_2\tilde{\bm{p}},\quad
    (\bm{g}\bm{q})_x=2c_2\bm{q}-2c_1\tilde{\bm{q}}.
\end{align}
The $t$-part consists of systems \eqref{DNLS-system-NC} and \eqref{DNLS-BT-t-part} together with the the evolution of the auxiliary function
\begin{equation}\label{g-evol-DNLS}
   {\bm g}_t={\bm g}\tilde{{\bm q}}({\bm p}_x-4{\bm p}{\bm q}{\bm p})+(\tilde{\bm{q}}_x+4\tilde{{\bm q}}\tilde{{\bm p}}\tilde{{\bm q}})\tilde{{\bm p}}{\bm g}-2c_2({\bm q} {\bm p}-\tilde{{\bm q}}\tilde{{\bm p}}).
\end{equation}
\end{subequations}
\end{enumerate}
\end{theorem}

For the first integrals of the B\"acklund transformations we have the following.

\begin{proposition}
  Let $\left[\bm f,\bm p  \tilde{\bm q}\right]=\left[\bm g,\bm q  \tilde{\bm p}\right]=0$.  Then, the system of differential equations \eqref{BT-DNLS-1} has the following first integral:
$$
\partial_x(\bm{f}\bm{p}\tilde{\bm{q}}\bm{f}-c_2\bm{f})=0.
$$
\end{proposition}
\begin{proof}
Indeed,
\begin{align*}
\partial_x(\bm{f}\bm{p}\tilde{\bm{q}}\bm{f}-c_2\bm{f})
&=(\bm{fp})_x\bm{\tilde{q}}\bm{f}+\bm{f}\bm{p}(\tilde{\bm{q}}\bm{f})_x-c_2\bm{f}_x
\stackrel{\eqref{BT-DNLS-1}}{=}2c_1\left[\bm f,\bm p  \tilde{\bm q}\right],
\end{align*}
whence it follows that $\partial_x(\bm{f}\bm{p}\tilde{\bm{q}}\bm{f}-c_2\bm{f})=0$, if $\left[\bm f,\bm p  \tilde{\bm q}\right]=0$. Similarly,
\begin{align*}
\partial_x(\bm{g}\bm{q}\tilde{\bm{p}}\bm{g}-c_1\bm{g})
&=(\bm{gq})_x\bm{\tilde{p}}\bm{g}+\bm{g}\bm{q}(\tilde{\bm{p}}\bm{g})_x-c_1\bm{g}_x
\stackrel{\eqref{BT-DNLS-2}}{=}2c_2\left[\bm g,\bm q  \tilde{\bm p}\right],
\end{align*}
which implies $\partial_x(\bm{g}\bm{q}\tilde{\bm{p}}\bm{g}-c_1\bm{g})=0$, if $\left[\bm g,\bm q  \tilde{\bm p}\right]=0$.
\end{proof}

For a semi-discretisation, we introduce a discrete variable $n\in\mathbb{Z}$, and we shall use it to place on the line the various solutions of system \eqref{DNLS-system-NC}, $\{{\bm p}_n,{\bm q}_n\}_{n\in\mathbb{Z}}$, where for $n=0$ we have the original solution of \eqref{DNLS-system-NC}, $({\bm p}_0,{\bm q}_0)\equiv ({\bm p},{\bm q})$, whereas the shifted $({\bm p}_1,{\bm q}_1)$ will be the solution obtained after the action of the Darboux matrix, namely $({\bm p}_0,{\bm q}_0)\equiv (\tilde{{\bm p}},\tilde{{\bm q}})$.

In the case where the parameters in either the Darboux matrix \eqref{DT-DNLS-1} or matrix \eqref{DT-DNLS-2} are zero, i.e. $c_1=c_2=0,$ then we obtain a noncommutative modified Volterra equation. For instance, from the B\"acklund transformation \eqref{BT-DNLS-1}, and according to the `discrete' notation we just introduced, it follows that 
$$
({\bm f}{\bm p})_x=({\bm q}_1{\bm f})=0.
$$
The relations ${\bm f}{\bm p}=1$ and ${\bm q}_1{\bm f}=1$ are solutions of the above differential relations and they imply $\bm{q}_1={\bm p}^{-1}$. The latter is equivalent to $\bm{q}\equiv\bm{q}_0=\bm{p}^{-1}_{-1}$. The Darboux matrix \eqref{DT-DNLS-1} becomes
\begin{equation}\label{DT-DNLS-1-deg}
{\rm M} =\lambda^2\begin{pmatrix} \bm{p}^{-1} & 0\\ 0 & 0\end{pmatrix}+\lambda \begin{pmatrix} 0 & 1\\ 1 & 0 \end{pmatrix}.
\end{equation}

Moreover, substituting $\bm{f}=\bm{p}^{-1}$ to the first equation of the B\"acklund transformation \eqref{BT-DNLS-1}:
$$
\bm{f}_x=2(\bm{p}_1\bm{q}_1\bm{f}-\bm{f}\bm{p}\bm{q}),
$$
and using $(\bm{p}^{-1})_x=-\bm{p}^{-1}\bm{p}_x\bm{p}^{-1}$, we obtain
\begin{equation}\label{mod-Vol}
    \bm{p}_x=2\bm{p}(\bm{p}_1\bm{p}^{-1}_0-\bm{p}^{-1}_{-1})\bm{p},
\end{equation}
which is the noncommutative avatar of the modified Volterra equation. The noncommutative modified Volterra equation was also obtained in \cite{Peroni-Wang} from linear in $\lambda$ NLS-type Darboux transformation. That is, the derivation from the B\"acklund transformation \eqref{BT-DNLS-1} was expected, since for $c_1=c_2=0$, the associated Darboux matrix \eqref{DT-BT-DNLS-1} becomes linear in $\lambda$.

\subsection{Deformation of the derivative NLS}
Now we find the Darboux matrix for the noncommutative deformation of the derivative NLS \eqref{system-def-DNLS}. The relations defining the Darboux matrix ${\rm M}$ are again given by \eqref{M-equation}, but now the Lax pair $({\rm U},{\rm V})$ is defined by \eqref{Lax-DDNLS-NC}.

We seek a Darboux matrix of the form
$$
{\rm M}=\lambda^2 {\rm M}^2 + \lambda {\rm M}^1 + {\rm M}^0  + \lambda^{-1} {\rm M}^{-1} + \lambda^{-2} {\rm M}^{-2}.
$$

It is necessary to take into account that the matrices ${\rm U}$ and ${\rm V}$ possess the symmetries
\begin{align*}
&{\rm U}(\lambda)= \sigma_3 {\rm U}(-\lambda)\sigma_3,\qquad  {\rm U}(\lambda)= \sigma_1 {\rm U}\left(\frac{1}{\lambda}\right)\sigma_1,\\
&{\rm V}(\lambda)= \sigma_3 {\rm V}(-\lambda)\sigma_3,\qquad  {\rm V}(\lambda)= \sigma_1 {\rm V}\left(\frac{1}{\lambda}\right)\sigma_1.
\end{align*}

The Darboux matrix must possess the same symmetries, i.e.
$$
{\rm M}(\lambda)= \sigma_3 {\rm M}(-\lambda)\sigma_3,\qquad  {\rm M}(\lambda)= \sigma_1 {\rm M}\left(\frac{1}{\lambda}\right)\sigma_1.
$$

As a consequence, the Darboux matrix must have the form
\begin{equation}\label{M-gen-form-D2}
\rm{M}=\lambda^2\begin{pmatrix}\bm\alpha&0\\0&\bm\beta\end{pmatrix}+\lambda\begin{pmatrix}0&\bm\kappa\\ \bm\delta&0\end{pmatrix}+ \begin{pmatrix}\bm\mu&0\\0&\bm\mu\end{pmatrix}+\lambda^{-1}\begin{pmatrix}0&\bm\delta\\ \bm\kappa&0\end{pmatrix}+\lambda^{-2}\begin{pmatrix}\bm\beta&0\\0&\bm\alpha\end{pmatrix}.
\end{equation}

We have the following theorem.

\begin{theorem}\label{DT-BT-DDNLS-th}
All possible Darboux transformations associated with the Lax pair \eqref{Lax-DDNLS-NC} of the noncommutative deformation of the derivative NLS \eqref{system-def-DNLS}, with such dependence on the spectral parameter and $\rank {\rm M}^2=1$, fall into one of the following cases.
\begin{enumerate}
    \item The Darboux matrix
    \begin{equation}\label{Darboux-DDNLS-NC}
     {\rm M}=\lambda^2
\begin{pmatrix}\bm f&0\\0&0\end{pmatrix}
+\lambda
\begin{pmatrix}0&\bm f\bm p\\ \tilde{\bm q}\bm f&0\end{pmatrix}
+\begin{pmatrix}\bm f\bm g&0\\0&\bm f\bm g\end{pmatrix}+\lambda^{-1}
\begin{pmatrix}0&\tilde{\bm q}\bm f\\ \bm f\bm p&0\end{pmatrix}
+\lambda^{-2}
\begin{pmatrix}0&0\\0&\bm f\end{pmatrix},
\end{equation}
and the corresponding B\"acklund transformation
\begin{subequations}\label{BT-DDNLS-NC}
\begin{align}
\bm f_x &= 2\bigl(\bm f\bm p\,\bm q-\tilde{\bm p}\tilde{\bm q}\bm f\bigr),\label{BT-DDNLS-NC-a}\\
(\bm f\bm p)_x &=2\bigl(\bm f\bm g\bm p+\bm f\bm q-\tilde{\bm q}\bm f-\tilde{\bm p}\bm f\bm g\bigr),\label{BT-DDNLS-NC-b}\\
(\tilde{\bm q}\bm f)_x&=2\bigl(\bm f\bm p-\tilde{\bm q}\bm f\bm g+\bm f\bm g\bm q-\tilde{\bm p}\bm f\bigr),\label{BT-DDNLS-NC-c}\\
(\bm f\bm g)_x&=2\bigl(\tilde{\bm q}\bm f\bm q+\bm f\bm p^2-\tilde{\bm p}\bm f\bm p-\tilde{\bm q}^2\bm f\bigr),\label{BT-DDNLS-NC-d}
\end{align}
\end{subequations}
while the essential part of its $t$-evolution is given by
\begin{subequations}\label{BT-DDNLS-NC-t-essential}
\begin{align}
\bm f_t&=
\tilde{\bm A}\tilde{\bm q}\bm f-\bm f\bm p\,\bm B
+\tilde\alpha\,\bm f\bm g-\bm f\bm g\,\alpha
+\tilde\beta\,\bm f-\bm f\beta
+2\tilde{\bm q}\bm f\bm q-2\tilde{\bm p}\bm f\bm p,\label{BT-DDNLS-NC-t-essential-a}\\
(\bm f\bm g)_t&=
\tilde{\bm A}\bm f\bm p-\bm f\bm p\,\bm A
+\tilde{\bm B}\tilde{\bm q}\bm f-\tilde{\bm q}\bm f\,\bm B
+\tilde\beta\,\bm f\bm g-\bm f\bm g\,\beta
+\tilde\gamma\,\bm f-\bm f\gamma,\label{BT-DDNLS-NC-t-essential-b}
\end{align}
\end{subequations}
where $S=\bm q\bm p+\bm p\bm q$, $A=\bm p_x+2\bm p\bm S$ and $B=-\bm q_x+2\bm q\bm S$, as defined earlier,
while $\alpha,\beta,\gamma$ and $\tilde\alpha,\tilde\beta,\tilde\gamma$ are the diagonal potentials entering the matrix ${\rm V}$ in \eqref{def-DNLS-V-diag-corr}, corresponding respectively to $(\bm p,\bm q)$ and $(\tilde{\bm p},\tilde{\bm q})$. In particular, since $\mathcal R$ is a division ring, \eqref{BT-DDNLS-NC-t-essential-a} and \eqref{BT-DDNLS-NC-t-essential-b} determine $\bm g_t$ by
\[
\bm g_t=\bm f^{-1}\Bigl((\bm f\bm g)_t-\bm f_t\bm g\Bigr).
\]

\item The Darboux matrix
    \begin{equation}\label{Darboux-DDNLS-NC-2}
    {\rm M}=\lambda^2
\begin{pmatrix}0&0\\0&\bm h\end{pmatrix}-\lambda
\begin{pmatrix}0&\tilde{\bm p}\bm h\\ \bm h\bm q&0\end{pmatrix}
+\begin{pmatrix}\bm h\bm g&0\\0&\bm h\bm g\end{pmatrix}-\lambda^{-1}\begin{pmatrix}0&\bm h\bm q\\ \tilde{\bm p}\bm h&0\end{pmatrix}+\lambda^{-2}\begin{pmatrix}\bm h&0\\0&0\end{pmatrix},
\end{equation}
and the corresponding B\"acklund transformation
\begin{subequations}\label{BT-DDNLS-NC-2}
\begin{align}
\bm h_x&=2\bigl(\tilde{\bm q}\tilde{\bm p}\bm h-\bm h\bm q\bm p\bigr),\label{BT-DDNLS-NC-a-2}\\
(\tilde{\bm p}\bm h)_x&=2\bigl(\tilde{\bm p}\bm h\bm g+\tilde{\bm q}\bm h-\bm h\bm g\,\bm p-\bm h\bm q\bigr),\label{BT-DDNLS-NC-b-2}\\
(\bm h\bm q)_x&=2\bigl(\tilde{\bm p}\bm h+\tilde{\bm q}\bm h\bm g
-\bm h\bm g\,\bm q-\bm h\bm p\bigr),\label{BT-DDNLS-NC-c-2}\\
(\bm h\bm g)_x&=2\bigl(\tilde{\bm p}^2\bm h+\tilde{\bm q}\bm h\bm q-\bm h\bm q^2-\tilde{\bm p}\bm h\bm p\bigr),\label{BT-DDNLS-NC-d-2}
\end{align}
\end{subequations}
while the essential part of its $t$-evolution is given by
\begin{subequations}\label{BT-DDNLS-NC-2-t-essential}
\begin{align}
\bm h_t&=
\bm h\bm q\,\bm A-\tilde{\bm B}\,\tilde{\bm p}\bm h
+\tilde\gamma\,\bm h\bm g-\bm h\bm g\,\gamma
+\tilde\beta\,\bm h-\bm h\beta
+2\tilde{\bm q}\bm h\bm q-2\tilde{\bm p}\bm h\bm p,\label{BT-DDNLS-NC-2-t-essential-a}\\
(\bm h\bm g)_t&=
\tilde\alpha\,\bm h-\bm h\alpha
+\tilde{\bm A}\,\tilde{\bm p}\bm h-\tilde{\bm p}\bm h\,\bm A
+\tilde{\bm B}\,\bm h\bm q-\bm h\bm q\,\bm B
+\tilde\beta\,\bm h\bm g-\bm h\bm g\,\beta,\label{BT-DDNLS-NC-2-t-essential-b}
\end{align}
\end{subequations}
and, since $\mathcal R$ is a division ring,
\[
\bm g_t=\bm h^{-1}\Bigl((\bm h\bm g)_t-\bm h_t\bm g\Bigr).
\]
\end{enumerate}
\end{theorem}

\begin{proof}
Let
\[
{\rm P}=\begin{pmatrix}0&\bm p\\ \bm q&0\end{pmatrix},
\qquad
{\rm Q}=\begin{pmatrix}0&\bm q\\ \bm p&0\end{pmatrix}.
\]
The equation $\partial_x{\rm M}+{\rm M}{\rm U}-\tilde{\rm U}{\rm M}=0$ from \eqref{M-equation} is equivalent to the following system of polynomial equations:
\begin{subequations}\label{DDNLS-pol-eq}
    \begin{align}
        &\lambda^{4}:\quad [\sigma_{3},{\rm M}^2]=0,\label{DDNLS-pol-eq-a}\\
        &\lambda^{3}:\quad  [\sigma_{3},{\rm M}^1]+2\tilde{{\rm P}}{\rm M}^2-2{\rm M}^2{\rm P}=0,\label{DDNLS-pol-eq-b}\\
        &\lambda^{2}:\quad ({\rm M}^2)_x+[\sigma_{3},{\rm M}^0]+2(\tilde{{\rm P}}{\rm M}^1-{\rm M}^1{\rm P})=0,\label{DDNLS-pol-eq-c}\\
        &\lambda^{1}:\quad  ({\rm M}^1)_x+[\sigma_{3},{\rm M}^{-1}]+2(\tilde{{\rm P}}{\rm M}^0-{\rm M}^0{\rm P})+2(\tilde{{\rm Q}}{\rm M}^2-{\rm M}^2{\rm Q})=0,\label{DDNLS-pol-eq-d}\\
        &\lambda^0:\quad ({\rm M}^0)_x+[\sigma_{3},{\rm M}^{-2}]+2(\tilde{{\rm P}}{\rm M}^{-1}-{\rm M}^{-1}{\rm P})+2(\tilde{{\rm Q}}{\rm M}^1-{\rm M}^1{\rm Q})=0.\label{DDNLS-pol-eq-e}
    \end{align}
\end{subequations}

Equation \eqref{DDNLS-pol-eq-a} is identically satisfied. Equation \eqref{DDNLS-pol-eq-b} is equivalent to
\begin{equation}\label{kappa-delta}
\bm\kappa=\bm\alpha\bm p-\tilde{\bm p}\bm\beta,\qquad
\bm\delta=\tilde{\bm q}\bm\alpha-\bm\beta\bm q,
\end{equation}
whereas equation \eqref{DDNLS-pol-eq-c} is equivalent to
\begin{equation}\label{alpha-beta-x}
\bm\alpha_x=2(\bm\kappa\bm q-\tilde{\bm p}\bm\delta),
\qquad
\bm\beta_x=2(\bm\delta\bm p-\tilde{\bm q}\bm\kappa).
\end{equation}
Moreover, equation \eqref{DDNLS-pol-eq-d} is equivalent to
\begin{equation}\label{kappa-delta-x}
\bm\kappa_x=2\bigl(
-\bm\delta-\tilde{\bm p}\bm\mu+\bm\mu\bm p
-\tilde{\bm q}\bm\beta+\bm\alpha\bm q
\bigr),\qquad
\bm\delta_x=2\bigl(
\bm\kappa-\tilde{\bm q}\bm\mu+\bm\mu\bm q
-\tilde{\bm p}\bm\alpha+\bm\beta\bm p
\bigr).
\end{equation}
Finally, equation \eqref{DDNLS-pol-eq-e} is equivalent to
\begin{equation}\label{mu-nu-x}
\bm\mu_x=2(\bm\delta\bm q+\bm\kappa\bm p
-\tilde{\bm p}\bm\kappa-\tilde{\bm q}\bm\delta).
\end{equation}

Since $\rank {\rm M}^2=1$, one of the quantities $\bm\alpha$ or $\bm\beta$ must vanish. Consider the case $\bm\alpha=\bm f\neq 0$ and $\bm\beta=0$. Without loss of generality, we set $\bm\mu=\bm f\bm g$. Then, according to \eqref{kappa-delta}, we obtain
\[
\bm\kappa=\bm f\bm p,\qquad \bm\delta=\tilde{\bm q}\bm f.
\]
Substituting these expressions into \eqref{alpha-beta-x}, \eqref{kappa-delta-x}, and \eqref{mu-nu-x}, we arrive at the $x$-part \eqref{BT-DDNLS-NC}.

To obtain the essential $t$-part, we substitute the Darboux matrix \eqref{Darboux-DDNLS-NC} into
\[
\partial_t{\rm M}+{\rm M}{\rm V}-\tilde{\rm V}{\rm M}=0,
\]
expand in powers of $\lambda$, and equate the coefficients of equal powers to zero. The coefficients of $\lambda^{\pm 6},\lambda^{\pm 5},\lambda^{\pm 4},\lambda^{\pm 3}$ are identically satisfied or reproduce the $x$-part relations already obtained, whereas the coefficients of $\lambda^2$ and $\lambda^0$ yield \eqref{BT-DDNLS-NC-t-essential-a} and \eqref{BT-DDNLS-NC-t-essential-b}, respectively.

The second case, corresponding to $\bm\alpha=0$ and $\bm\beta=\bm h\neq 0$, is treated in the same way. This gives the Darboux matrix \eqref{Darboux-DDNLS-NC-2}, the $x$-part \eqref{BT-DDNLS-NC-2}, and the essential $t$-part \eqref{BT-DDNLS-NC-2-t-essential}.
\end{proof}

\begin{remark}
The second Darboux family \eqref{Darboux-DDNLS-NC-2}--\eqref{BT-DDNLS-NC-2-t-essential} is obtained from the first one by the involution
\[
{\rm M}(\lambda)\mapsto \sigma_1{\rm M}(\lambda^{-1})\sigma_1,
\]
together with the exchange
\[
(\bm p,\bm q)\leftrightarrow (\bm q,\bm p),\qquad
(\tilde{\bm p},\tilde{\bm q})\leftrightarrow (\tilde{\bm q},\tilde{\bm p}).
\]
Therefore, the two cases are equivalent up to the natural symmetry of the Lax pair, and the second one does not yield an essentially new B\"acklund transformation.
\end{remark}

Now, in order to write the B\"acklund transformations explicitly in terms of $({\bm f}_x, {\bm p}_x, \tilde{\bm q}_x,{\bm g}_x)$, we need to assume some commutativity relations. In particular, we have the following.

\begin{proposition}
    Let $\left[{\bm f},\tilde{{\bm p}}\right]=\left[{\bm f},\tilde{{\bm q}}\right]=0$. Then, the B\"acklund transformation \eqref{BT-DDNLS-NC} can be written as:
    \begin{subequations}\label{BT-DDNLS-NC-epxressed}
\begin{align}
\bm f_x &= 2{\bm f}\bigl(\bm p\,\bm q-\tilde{\bm p}\tilde{\bm q}\bigr),\label{BT-DDNLS-NC-epxressed-a}\\
{\bm p}_x &=2\bigl(\bm g\bm p+\bm q-\tilde{\bm q}-\tilde{\bm p}\bm g\bigr)-{\bm f}^{-1}{\bm f}_x{\bm p},\label{BT-DDNLS-NC-epxressed-b}\\
{\tilde{\bm q}}_x&=2\bigl(\bm p-\tilde{\bm q}\bm g+\bm g\bm q-\tilde{\bm p}\bigr)-\tilde{{\bm q}}{\bm f}^{-1}{\bm f}_x,\label{BT-DDNLS-NC-epxressed-c}\\
{\bm g}_x&=2\bigl(\tilde{\bm q}\bm q+\bm p^2-\tilde{\bm p}\bm p-\tilde{\bm q}^2\bigr)-{\bm f}^{-1}{\bm f}_x{\bm g}.\label{BT-DDNLS-NC-epxressed-d}
\end{align}
\end{subequations}
\end{proposition}
\begin{proof}
    Set $\alpha={\bm f}^{-1}{\bm f}_x$. Then, ${\bm f}_x={\bm f}\alpha$, equation \eqref{BT-DDNLS-NC-a} can be written as \eqref{BT-DDNLS-NC-epxressed-a}, and the latter is equivalent to $2\bigl(\bm p\,\bm q-\tilde{\bm p}\tilde{\bm q}\bigr)=\alpha$. Equation \eqref{BT-DDNLS-NC-b} is equivalent to:
$$
\bm f_x\bm p+\bm f\bm p_x =2{\bm f}\bigl(\bm g\bm p+\bm q-\tilde{\bm q}-\tilde{\bm p}\bm g\bigr)
$$
and after multiplying the above from the left by ${\bm f}^{-1}$, we obtain \eqref{BT-DDNLS-NC-epxressed-b}. In a similar way, we obtain \eqref{BT-DDNLS-NC-epxressed-c} and \eqref{BT-DDNLS-NC-epxressed-d}.
\end{proof}

For the first integrals of system \eqref{BT-DDNLS-NC-epxressed}, we need to assume some extra commutativity conditions. Specifically, we have the following. 

\begin{proposition}
    Let ${\bm f}, {\bm g}\in Z(\mathcal{R})$. Then, we have the following.
    \begin{enumerate}
        \item If $\left[\tilde{{\bm p}},{\bm p}\right]=\left[\tilde{{\bm q}},{\bm q}\right]$, then $\Phi_1=\bm{f}^2\left[2{\bm f}^{-1}{\bm f}_x({\bm g}-{\bm p}\tilde{{\bm q}})+({\bm g}-{\bm p}\tilde{{\bm q}})_x\right]$ is a first integral of system \eqref{BT-DDNLS-NC-epxressed}.
        \item If  $\left[\tilde{{\bm p}},{\bm p}\right]=\left[\tilde{{\bm q}},{\bm q}\right]$ and $\left[\tilde{{\bm p}},\tilde{{\bm q}}\right]=\left[{\bm p},{\bm q}\right]$, then $\Phi_2={\bm f}^2\left({\bm g}^2+1-{\bm p}^2-\tilde{{\bm q}}^2\right)$ is a first integral of system \eqref{BT-DDNLS-NC-epxressed}.
    \end{enumerate}
\end{proposition}
\begin{proof}
    Using equations of system \eqref{BT-DDNLS-NC-epxressed}, one can prove that:
    $$
    \Phi_{1,x}=2{\bm f}^2(\left[\tilde{\bm{q}},\bm{q}\right]-\left[\tilde{\bm{p}},\bm{p}\right]),\quad \Phi_{2,x}=2{\bm f}^2\left\{\bm{g}\left(\left[\tilde{\bm{q}},\bm{q}\right]-\left[\tilde{\bm{p}},\bm{p}\right]\right)+\left[\bm{p},\bm{q}\right]-\left[\tilde{\bm{p}},\tilde{\bm{q}}\right]\right\}.
    $$
\end{proof}

\section{Noncommutative integrable discretisations of NLS systems}\label{NCDISNLS}
In this section, we first briefly recall the integrable discretisations of the noncommutative NLS system \eqref{NLS-eq-NC} derived in \cite{KRX} so that the section is self-contained. Then, we construct integrable discretisations for the noncommutative derivative NLS system \eqref{DNLS-system-NC} and the noncommutative deformation of the derivative NLS system \eqref{system-def-DNLS}.

\subsection{Discretisation of the noncommutative NLS system}
In \cite{KRX} we employed the Darboux matrix \eqref{DT-10} in a discrete setting, i.e. ${\rm{M}}={\rm{M}}(\bm{f}_{00},\bm{p}_{00},\bm{q}_{10})$, and we showed that the discrete zero-curvature condition
$$
{\rm{M}}(\bm{f}_{01},\bm{p}_{01},\bm{q}_{11}){\rm{M}}(\bm{g}_{00},\bm{p}_{00},\bm{q}_{01})={\rm{M}}(\bm{g}_{10},\bm{p}_{10},\bm{q}_{11}){\rm{M}}(\bm{f}_{00},\bm{p}_{00},\bm{q}_{10}),
$$
is equivalent to the noncommutative discrete integrable system
\begin{subequations}\label{NLS-dis-sys}
\begin{align}
    & \bm{f}_{01}+\bm{g}_{00}=\bm{g}_{10}+\bm{f}_{00},\label{NLS-dis-sys-a}\\
    & \bm{f}_{01}\bm{g}_{00}+\bm{p}_{01}\bm{q}_{01}=\bm{g}_{10}\bm{f}_{00}+\bm{p}_{10}\bm{q}_{10},\label{NLS-dis-sys-b}\\
    & \bm{f}_{01}\bm{p}_{00}+\bm{p}_{01}=\bm{g}_{10}\bm{p}_{00}+\bm{p}_{10},\label{NLS-dis-sys-c}\\
    & \bm{q}_{11}\bm{g}_{00}+\bm{q}_{01}=\bm{q}_{11}\bm{f}_{00}+\bm{q}_{10}\label{NLS-dis-sys-d}.
\end{align}
\end{subequations}
We showed that, under the conditions $[\bm{f}_{00},\bm{p}_{00} \bm{q}_{11}] = [\bm{g}_{00},\bm{p}_{00} \bm{q}_{11}] =0$, system \eqref{NLS-dis-sys} admits the first integrals:
$$
({\cal{T}}-1)\left(\bm{f}_{00}-\bm{p}_{00}\bm{q}_{10}\right) = 0,\quad ({\cal{S}}-1)\left(\bm{g}_{00}-\bm{p}_{00}\bm{q}_{01}\right) = 0.
$$
Then, on the level-sets of the above first integrals, we restricted system \eqref{NLS-dis-sys} to the noncommutative Adler--Yamilov system:
$$
 \bm{p}_{01}=\bm{p}_{10}-\left(a- b\right)(1+\bm{p}\bm{q}_{11})^{-1}\bm{p} =\quad \bm{q}_{01}=\bm{q}_{10}+\left(a- b\right)\bm{q}_{11},
$$
which is the noncommutative version of system \eqref{Adler-Yamilov}.

\subsection{Discretisation of the noncommutative derivative NLS system}
For the discretisation of system \eqref{DNLS-system-NC}, we employ the Darboux matrix \eqref{DT-DNLS-1} in a discrete notation ${\rm M}={\rm M}(\bm{f}_{00},\bm{p}_{00},\bm{q}_{10};c_1,c_2)$ and also consider matrix ${\rm N}={\rm M}(\bm{g}_{00},\bm{p}_{00},\bm{q}_{01};1,1)$. 

We have the following. 

\begin{proposition}\label{sl2-nc-prop} The following system of partial difference equations
    \begin{subequations} \label{sl2-res-eq-nc}
\begin{align}
 &{\bm f}_{01}{\bm g}_{00}-{\bm g}_{10}{\bm f}_{00} = 0,\\
 & {\bm q}_{10}{\bm f}_{00}-{\bm q}_{11}{\bm f}_{01}  + c_1 {\bm q}_{11}{\bm g}_{10} - c_2 {\bm q}_{01}{\bm g}_{00}=0, \\
 &{\bm f}_{01}{\bm p}_{01} - {\bm f}_{00}{\bm p}_{00} - c_2 {\bm g}_{10}{\bm p}_{10} + c_1 {\bm g}_{00}{\bm p}_{00} = 0,\\
 &{\bm f}_{01}-{\bm f}_{00}- c_1 ({\bm g}_{10}-{\bm g}_{00}) - {\bm g}_{10} {\bm p}_{10}{\bm q}_{10}{\bm f}_{00} + {\bm f}_{01}{\bm p}_{01}{\bm q}_{01} {\bm g}_{00} = 0,
\end{align}
\end{subequations}
is integrable with Lax pair
\begin{equation}\label{Lax-NC-Dis-DNLS}
    \rm{\Psi}_{10}={\rm M}(\bm{f}_{00},\bm{p}_{00},\bm{q}_{10};c_1,c_2)\rm{\Psi},\quad \rm{\Psi}_{01}={\rm M}(\bm{g}_{00},\bm{p}_{00},\bm{q}_{01};1,1)\rm{\Psi},
\end{equation}
where ${\rm M}(\bm{f}_{00},\bm{p}_{00},\bm{q}_{10};c_1,c_2)=\begin{pmatrix}
    \lambda^2 {\bm f}_{00}+c_1 & \lambda {\bm f}{\bm p}\\
    \lambda {\bm q}_{10} & c_2
\end{pmatrix}$.
\end{proposition}
\begin{proof}
    It can be readily shown that system \eqref{sl2-res-eq-nc} is equivalent to the compatibility condition of system \eqref{Lax-NC-Dis-DNLS}, 
    $$
    {\rm M}(\bm{f}_{01},\bm{p}_{01},\bm{q}_{11};c_1,c_2){\rm M}(\bm{g}_{00},\bm{p}_{00},\bm{q}_{01};1,1)={\rm M}(\bm{g}_{10},\bm{p}_{10},\bm{q}_{11};1,1){\rm M}(\bm{f}_{00},\bm{p}_{00},\bm{q}_{10};c_1,c_2),
    $$
    for any $\lambda\in\mathbb{C}$.
\end{proof}

System \eqref{sl2-res-eq-nc} is a vertex-bond system where the fields ${\bm p}_{ij}$ and ${\bm q}_{ij}$, $i,j\in\{0,1\}$, are placed on the vertices of the square, while ${\bm f}_{ij}$ and ${\bm g}_{ij}$, $i,j\in\{0,1\}$, lie on the edges of the square, as in Figure \ref{fig-ivp}.

\begin{figure}[ht]
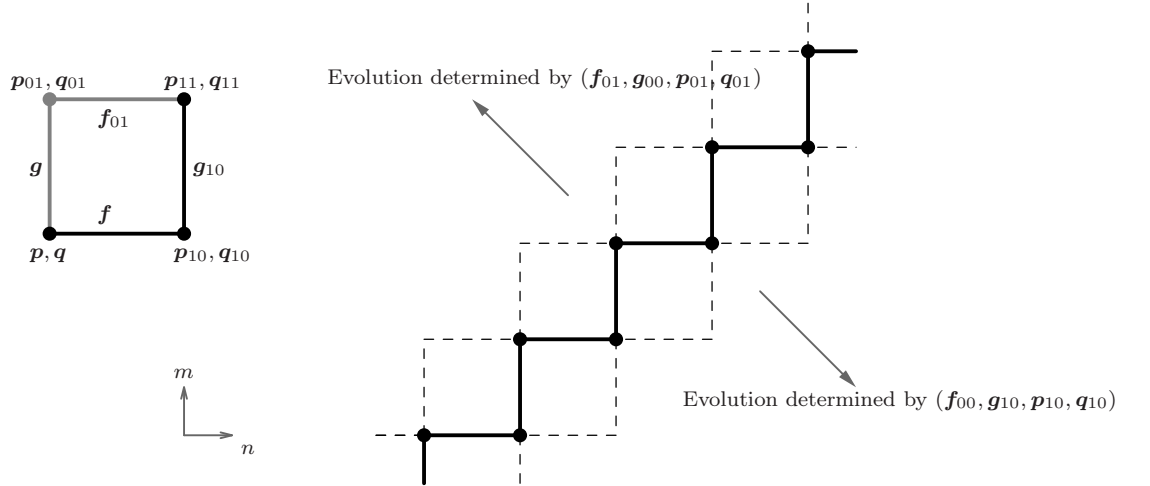

\centertexdraw{
\setunitscale 0.5
\move(-4.5 -2) \linewd 0.02 \setgray 0.4 \arrowheadtype t:V \arrowheadsize l:.12 w:.06 \avec(-4.5 -1.5) 
\move(-4.5 -2) \arrowheadtype t:V  \avec(-4 -2)
\arrowheadsize l:.20 w:.10
\move(-.5 .5) \linewd 0.02 \setgray 0.4 \arrowheadtype t:F \avec(-1.5 1.5) 
\move(1.5 -.5) \linewd 0.02 \setgray 0.4 \arrowheadtype t:F \avec(2.5 -1.5) 
\setgray 0.5
\linewd 0.04 \move (-4.5 1.5)  \lvec (-5.9 1.5) \lvec (-5.9 .1)
\move (-5.9 1.5) \fcir f:0.5 r:0.075
\htext (-6.1 .7) {\scriptsize{$\bm{g}$}}
\htext (-4.4 .7) {\scriptsize{$\bm{g}_{10}$}}
\setgray 0.0
\linewd 0.04 \move (-2 -2.5) \lvec (-2 -2) \lvec (-1 -2) \lvec (-1 -1) \lvec (0 -1) \lvec (0 0) \lvec (1 0) \lvec(1 1) \lvec (2 1) \lvec(2 2) \lvec(2.5 2)
\move (-5.9 .1) \lvec (-4.5 .1) \lvec (-4.5 1.5)  
\linewd 0.015 \lpatt (.1 .1 ) \move (-2 -2) \lvec (-2 -1) \lvec(-1 -1) \lvec (-1 0) \lvec (0 0) \lvec (0 1) \lvec(1 1) \lvec (1 2) \lvec (2 2) \lvec (2 2.5)
\move(-2.5 -2) \lvec(-2 -2) \move(-2.5 -2) \lvec(-2 -2)
\move (-1 -2.5) \lvec (-1 -2) \lvec(0 -2) \lvec(0 -1) \lvec(1 -1) \lvec(1 0) \lvec(2 0) \lvec(2 1) \lvec(2.5 1)
\move (-2 -2) \fcir f:0.0 r:0.075 \move (-1 -2) \fcir f:0.0 r:0.075
\move (-1 -1) \fcir f:0.0 r:0.075 \move (0 -1) \fcir f:0.0 r:0.075
\move (0 0) \fcir f:0.0 r:0.075 \move (1 0) \fcir f:0.0 r:0.075  
\move (1 1) \fcir f:0.0 r:0.075 \move (2 1) \fcir f:0.0 r:0.075
\move (2 2) \fcir f:0.0 r:0.075
\move (-5.9 .1) \fcir f:0.0 r:0.075 \move (-4.5 .1) \fcir f:0.0 r:0.075 \move (-4.5 1.5) \fcir f:0.0 r:0.075
\htext (-3.9 -2.2) {\scriptsize{$n$}}
\htext (-4.6 -1.4) {\scriptsize{$m$}}
\htext (-6.1 -.2) {\scriptsize{$\bm{p},\bm{q}$}}
\htext (-4.6 -.2) {\scriptsize{$\bm{p}_{10},\bm{q}_{10}$}}
\htext (-5.4 0.2) {\scriptsize{$\bm{f}$}}
\htext (-5.4 1.2) {\scriptsize{$\bm{f}_{01}$}}
\htext (-6.3 1.6) {\scriptsize{$\bm{p}_{01},\bm{q}_{01}$}}
\htext (-4.7 1.6) {\scriptsize{$\bm{p}_{11},\bm{q}_{11}$}}
\htext (-3 1.6) {{\scriptsize{Evolution determined by $(\bm{f}_{01},\bm{g}_{00},\bm{p}_{01},\bm{q}_{01})$}}}
\htext (.7 -1.75) {{\scriptsize{Evolution determined by $(\bm{f}_{00},\bm{g}_{10},\bm{p}_{10},\bm{q}_{10})$}}}
}
\caption{{Vertex-bond system. Initial value problem.}} \label{fig-ivp}
\end{figure}

For system \eqref{sl2-res-eq-nc} we can define the initial-value-problem on the staircase, as in Figure \ref{fig-ivp}. That is, we can solve either for $(\bm{f}_{01},\bm{g}_{00},\bm{p}_{01},\bm{q}_{01})$ or  $(\bm{f}_{00},\bm{g}_{10},\bm{p}_{10},\bm{q}_{10})$. If we solve for the first tetrad, after some standard calculations, we obtain a trivial solution
$$
\bm{f}_{01}=c_1\bm{g}_{10},\quad \bm{g}_{00}=\frac{1}{c_1}\bm{f}_{00},\quad \bm{p}_{01}=\frac{c_2}{c_1}\bm{p}_{10},\quad \bm{q}_{01}=\frac{c_1}{c_2}\bm{q}_{10},
$$
and a nontrivial solution, given by:
\begin{subequations}\label{nc-sl2-ivp}
    \begin{align}
        \bm{g}_{00}&=\bm{g}_{10}\bm{A}\bm{B}^{-1},\\
        \bm{f}_{01}&=\bm{f}_{00}\bm{B}\bm{A}^{-1},\\
        \bm{p}_{01}&=\bm{A}\bm{B}^{-1}\bm{f}_{00}^{-1}\left[\bm{g}_{10}(c_2\bm{p}_{10}-\bm{A}\bm{B}^{-1}\bm{p}_{00})+\bm{f}_{00}\bm{p}_{00}\right],\\
        \bm{q}_{01}&=(c_2\bm{g}_{10})^{-1}\left[\bm{f}_{00}(\bm{q}_{10}-\bm{B}\bm{A}^{-1}\bm{q}_{11})+c_1\bm{g}_{10}\bm{q}_{11}\right]\bm{B}\bm{A}^{-1},
    \end{align}
\end{subequations}
which defines the evolution towards the north-west direction of the  staircase.

Next, we employ the Darboux matrices \eqref{DT-DNLS-1} and Darboux matrix \eqref{DT-DNLS-1-deg} to construct another integrable equation. In particular, we consider the system of linear equations
\begin{equation}\label{Lax-6point}
    \Psi_{10}=\begin{pmatrix}\lambda^2\bm{p}_{00}^{-1} & \lambda\\ \lambda &  0 \end{pmatrix}\Psi_{00},\quad  \Psi_{01}=\begin{pmatrix}\lambda^2\bm{g}_{00}+1 & \lambda\bm{g}_{00}\bm{p}\\ \lambda\bm{q}_{01}\bm{g}_{00} &  1 \end{pmatrix}\Psi_{00}.
\end{equation}

The compatibility condition of the overdetermined system \eqref{Lax-6point} is equivalent to the system 
\begin{subequations}\label{6point-sys}
    \begin{align}
        &\bm{p}_{01}^{-1}\bm{g}_{00}=\bm{g}_{10}\bm{p}_{00}^{-1},\label{6point-sys-a}\\
&\bm{p}_{01}^{-1}+\bm{q}_{01}\bm{g}_{00}=\bm{p}_{00}^{-1}+\bm{g}_{10}\bm{p}_{10},\label{6point-sys-b}\\
&\bm{g}_{00}=\bm{q}_{11}\bm{g}_{10}\bm{p}_{00}^{-1}.\label{6point-sys-c}
    \end{align}
\end{subequations}
Equations \eqref{6point-sys-a} and \eqref{6point-sys-c} can be written as $\bm{g}_{10}=\bm{p}_{01}^{-1}\bm{g}_{00}\bm{p}_{00}$ and $\bm{g}_{10}=\bm{q}_{11}^{-1}\bm{g}_{00}\bm{p}_{00}$, respectively. From the latter it follows that $\bm{p}_{01}=\bm{q}_{11}$ which is equivalent to $\bm{p}_{00}=\bm{q}_{10}$ or  $\bm{p}_{-11}=\bm{q}_{01}$. That is, \eqref{6point-sys-b} can be written as $\bm{p}_{01}^{-1}+\bm{p}_{-11}\bm{g}_{00}=\bm{p}_{00}^{-1}+\bm{p}_{01}^{-1}\bm{g}_{00}\bm{p}_{00}\bm{p}_{10}$. If $\left[\bm{g},\bm{p}_{-11}^{-1}\bm{p}_{01}^{-1}\right]$, then the former equation can be solved for $\bm{g}_{00}$:
$$
\bm{g}_{00}=\bm{p}_{-11}^{-1}(\bm{p}_{00}^{-1}-\bm{p}_{01}^{-1})(1-\bm{p}_{-11}^{-1}\bm{p}_{01}^{-1}\bm{p}_{00}\bm{p}_{10})^{-1}\bm{p}_{00}(1-\bm{p}_{01}^{-1}\bm{p}_{11}^{-1}\bm{p}_{10}\bm{p}_{20}).
$$

Now, equation \eqref{6point-sys-a}, after substituting $\bm{g}_{00}$ from the above expression, takes the form
\begin{equation}\label{6-point-NC}
    \bm{p}_{10}^{-1}-\bm{p}_{11}^{-1}=\bm{p}_{-11}^{-1}(\bm{p}_{00}^{-1}-\bm{p}_{01}^{-1})(1-\bm{p}_{-11}^{-1}\bm{p}_{01}^{-1}\bm{p}_{00}^{-1}\bm{p}_{10})^{-1}\bm{p}_{00}(1-\bm{p}_{01}^{-1}\bm{p}_{11}^{-1}\bm{p}_{10}\bm{p}_{20}).
\end{equation}
Equation \eqref{6-point-NC} is a noncommutative 6-point difference equation defined on the stencil of Figure \ref{6point-ivp}. If we assume that $\bm{p}_{ij}$, $i,j\in\{-1,0,1,2\}$, commute with each other, then the noncommutative difference equation \eqref{6-point-NC} reduces to the standard six-point formula (53) in \cite{SPS}.

Now we can define the initial-value problem on the double staircase, as shown in Figure \ref{6point-ivp}. Equation \eqref{6-point-NC} can be solved for $\bm{p}_{20}$. Indeed, after a couple of simple steps, we obtain:
\begin{equation}\label{6-point-NC-IVP}
    \bm{p}_{20}=\bm{p}_{10}^{-1}\bm{p}_{11}\bm{p}_{01}\left[1-(\bm{p}_{-11}^{-1}(\bm{p}_{00}^{-1}-\bm{p}_{01}^{-1})(1-\bm{p}_{-11}^{-1}\bm{p}_{01}^{-1}\bm{p}_{00}^{-1}\bm{p}_{10})^{-1}\bm{p}_{00})^{-1}(\bm{p}_{10}^{-1}-\bm{p}_{11}^{-1})\right].
\end{equation}
Equation \eqref{6-point-NC-IVP} defines the evolution of equation \eqref{6-point-NC} towards the south-east direction of the double staircase, as shown in Figure \ref{6point-ivp}. Note that the commutative version of  \eqref{6-point-NC}, namely equation (53) in \cite{SPS} which reads:
\begin{equation}\label{sixpointEq}
\frac{p_{01}-p}{p_{01}(p_{01}p_{-11}-pp_{10})}=\frac{p_{11}-p_{10}}{p_{10}(p_{11}p_{01}-p_{10}p_{20})},
\end{equation}
can be solved for $p_{-11}$ as well, whereas the noncommutative equation \eqref{6-point-NC} can be solved for $\bm{p}_{-11}$ only if additional commutativity conditions are assumed.

\begin{figure}[ht]
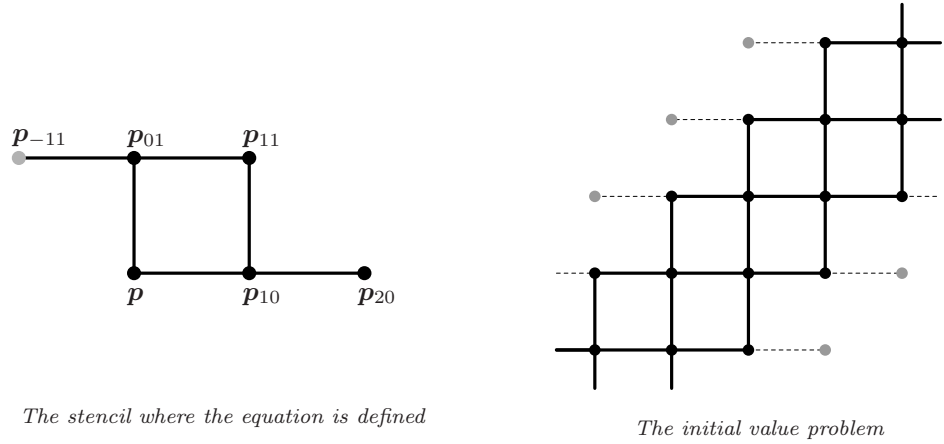

\centertexdraw{
\setunitscale 0.4
\setgray 0.0
\linewd 0.04 \move (-8 -1) \lvec (-6.5 -1) \lvec (-6.5 0.5) \lvec (-8 0.5) \lvec (-8 -1) \move (-8 0.5) \lvec (-9.5 0.5) \move (-6.5 -1) \lvec (-5 -1)
\move (-8 -1) \fcir f:0.0 r:0.095 \move (-6.5 -1) \fcir f:0.0 r:0.095
\move (-6.5 0.5) \fcir f:0.0 r:0.095 \move (-8 0.5) \fcir f:0.0 r:0.095
\move (-9.5 0.5) \fcir f:0.7 r:0.095 \move (-5 -1) \fcir f:0.0 r:0.095
\move (-2 -2.5) \lvec (-2 -2) \lvec (-1 -2) \lvec (-1 -1) \lvec (0 -1) \lvec (0 0) \lvec (1 0) \lvec(1 1) \lvec (2 1) \lvec(2 2) \lvec(2.5 2)
\linewd 0.04  \move (-2 -2) \lvec (-2 -1) \lvec(-1 -1) \lvec (-1 0) \lvec (0 0) \lvec (0 1) \lvec(1 1) \lvec (1 2) \lvec (2 2) \lvec (2 2.5) \move(-2.5 -2) \lvec(-2 -2) \move(-2.5 -2) \lvec(-2 -2)
\move (-1 -2.5) \lvec (-1 -2) \lvec(0 -2) \lvec(0 -1) \lvec(1 -1) \lvec(1 0) \lvec(2 0) \lvec(2 1) \lvec(2.5 1)
\linewd 0.001
\lpatt (.05 .05) \move (-2 0) \lvec (-1 0) \move (0 1) \lvec (-1 1) \move (0 2) \lvec (1 2) \move (-2.5 -1) \lvec (-2 -1) 
\move  (0 -2) \lvec (1 -2) \move (1 -1) \lvec (2 -1) \move (2 0) \lvec (2.5 0)
\lpatt  (1 1)
\move (-2 -2) \fcir f:0.0 r:0.075 \move (-2 -1) \fcir f:0.0 r:0.075
\move (-1 -2) \fcir f:0.0 r:0.075 \move (0 -2) \fcir f:0.0 r:0.075
\move (-1 -1) \fcir f:0.0 r:0.075 \move (-1 0) \fcir f:0.0 r:0.075 \move (0 -1) \fcir f:0.0 r:0.075 \move (1 -1) \fcir f:0.0 r:0.075
\move (0 0) \fcir f:0.0 r:0.075 \move (0 1) \fcir f:0.0 r:0.075 \move (1 0) \fcir f:0.0 r:0.075 \move (2 0) \fcir f:0.0 r:0.075 
\move (1 1) \fcir f:0.0 r:0.075 \move (2 1) \fcir f:0.0 r:0.075
\move (2 2) \fcir f:0.0 r:0.075 \move (1 2) \fcir f:0.0 r:0.075
\move (-2 0) \fcir f:0.6 r:0.075 
\move (-1 1) \fcir f:0.6 r:0.075 \move (0 2) \fcir f:0.6 r:0.075
\move (1 -2) \fcir f:0.6 r:0.075 \move (2 -1) \fcir f:0.6 r:0.075
\htext (-8.1 -1.4) {\small{$\bm{p}$}}
\htext (-6.6 -1.4) {\small{$\bm{p}_{10}$}}
\htext (-8.1 0.65) {\small{$\bm{p}_{01}$}}
\htext (-6.6 0.65) {\small{$\bm{p}_{11}$}}
\htext (-5.1 -1.4) {\small{$\bm{p}_{20}$}}
\htext (-9.6 0.65) {\small{$\bm{p}_{-11}$}}
\htext (-9.5 -3) {\scriptsize{\it{The stencil where the equation is defined}}}
\htext (-1.5 -3.15) {\scriptsize{\it{The initial value problem}}}
}
\caption{{\small{The stencil of six points and the initial value problem for equation \eqref{6-point-NC}}}} \label{6point-ivp}
\end{figure}

The commutative 6-point equation \eqref{sixpointEq} admits travelling wave solutions. Indeed, by making the ansatz $p(n,m)=F(an+bm)$, and substituting the latter into \eqref{sixpointEq}, it follows that $F(an+bm)$ is a solution if $b=2a$. That is, $p(n,m)=F(n+2m)$ is a solution of \eqref{sixpointEq}. A graph of such solutions in the case when $F(s)=1+2\sech^2\left[0.1 (s-10)\right]$ and $F(s)=1+0.8\tanh^2(s)$ is given in Figure \ref{Kink}.

\begin{figure}[ht]
\begin{center}
  \includegraphics{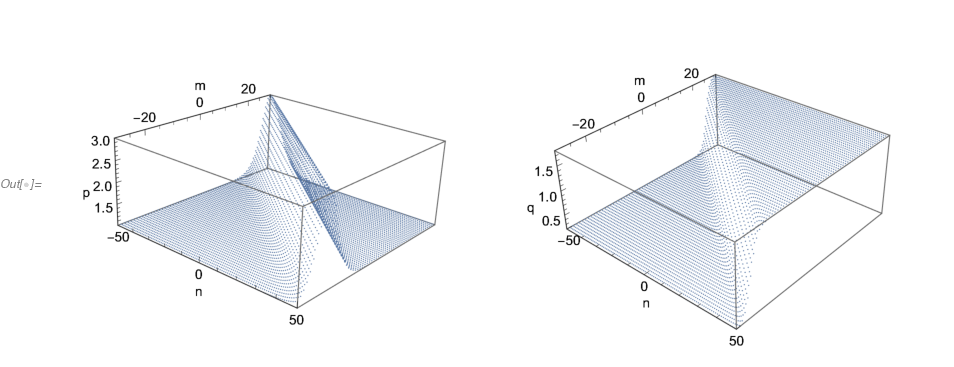} 
\end{center}
 \caption{Solutions of equation \eqref{sixpointEq}, given by  $p(n,m)=1+2\sech^2\left[0.1 (n+2m-10)\right]$ and $p(n,m)=1+0.8\tanh^2(n+2m)$.}
 \label{Kink}
\end{figure}

Now, let ${\rm A},{\rm B}\in \operatorname{Mat}_N(\mathbb C)$ satisfy ${\rm A}^2={\rm B}^2=I_N$, ${\rm A}{\rm B}\neq {\rm B}{\rm A}$. The following function
$$
\bm p(n,m)=(-1)^{m+\left\lfloor n/2\right\rfloor}
\begin{cases}
A, & n \ \mathrm{even},\\[1mm]
B, & n \ \mathrm{odd},
\end{cases}
$$
where $\left\lfloor x\right\rfloor$ denotes the largest integer $\leq x$, is a noncommutative solution of equation \eqref{6-point-NC-IVP}.

\subsection{Discretisation of the deformation of the derivative NLS}
For the discretisation of system \eqref{DNLS-system-NC}, we employ matrix \eqref{Darboux-DDNLS-NC}.

Specifically, we have the following.

\begin{proposition}\label{DDNS-dis-nc-prop}
    The following system of partial difference  equations
    \begin{subequations}\label{ddns-dis-nc}
        \begin{align}
        \bm{f}_{01}\bm{u}_{00}&=\bm{u}_{10}\bm{f}_{00},\label{ddns-dis-nc-a}\\
            \bm{f}_{01}(\bm{u}_{00}\bm{v}_{00}+\bm{g}_{01}\bm{u}_{00}+\bm{p}_{01}\bm{q}_{01}\bm{u}_{00})&=\bm{u}_{10}(\bm{f}_{00}\bm{g}_{00}+\bm{v}_{10}\bm{f}_{00}+\bm{p}_{10}\bm{q}_{10}\bm{f}_{00}),\\
            \bm{f}_{01}(\bm{g}_{01}\bm{u}_{00}\bm{v}_{00}+\bm{p}_{01}\bm{u}_{00}\bm{p}_{00})+\bm{q}_{11}\bm{f}_{01}\bm{q}_{01}\bm{u}_{00}&=\bm{u}_{10}(\bm{v}_{10}\bm{f}_{00}\bm{g}_{00}+\bm{p}_{10}\bm{f}_{00}\bm{p}_{00})+\bm{q}_{11}\bm{u}_{10}\bm{q}_{10}\bm{f}_{00},\\
            \bm{f}_{01}(\bm{q}_{01}\bm{u}_{00}+\bm{g}_{01}\bm{u}_{00}\bm{p}_{00}+\bm{p}_{01}\bm{u}_{00}\bm{v}_{00})&=\bm{u}_{10}(\bm{q}_{10}\bm{f}_{00}+\bm{v}_{10}\bm{f}_{00}\bm{p}_{00}+\bm{p}_{10}\bm{f}_{00}\bm{g}_{00}),\\
            \bm{f}_{01}(\bm{g}_{01}\bm{q}_{01}\bm{u}_{00}+\bm{p}_{01}\bm{u}_{00})+\bm{q}_{11}\bm{f}_{01}\bm{u}_{00}\bm{v}_{00}&=\bm{u}_{10}(\bm{v}_{10}\bm{q}_{10}\bm{f}_{00}+\bm{p}_{10}\bm{f}_{00})+\bm{q}_{11}\bm{u}_{10}\bm{f}_{00}\bm{g}_{00},
        \end{align}
    \end{subequations}
    is integrable with Lax pair
    \begin{equation}\label{Lax-NC-Dis-DDNLS}
    \rm{\Psi}_{10}={\rm M}(\bm{p}_{00},\bm{q}_{10},\bm{f}_{00},\bm{g}_{00})\rm{\Psi}_{00},\quad \rm{\Psi}_{01}={\rm M}(\bm{p}_{00},\bm{q}_{01},\bm{u}_{00},\bm{v}_{00})\rm{\Psi}_{00},
\end{equation}
    where ${\rm M}(\bm{p}_{00},\bm{q}_{10},\bm{f}_{00},\bm{g}_{00})=\begin{pmatrix}
        \lambda^2\bm{f}_{00}+\bm{f}_{00}\bm{g}_{00} & \lambda \bm{f}_{00}\bm{p}_{00}+\lambda^{-1}\bm{q}_{10}\bm{f}_{00}\\
        \lambda \bm{q}_{10}\bm{f}_{00}+\lambda^{-1}\bm{f}_{00}\bm{p}_{00} & \lambda^{-2}\bm{f}_{00}+\bm{f}_{00}\bm{g}_{00}
    \end{pmatrix}$.
\end{proposition}
\begin{proof}
    It can be shown by straightforward computation that system \eqref{ddns-dis-nc} is equivalent to the compatibility condition of system \eqref{Lax-NC-Dis-DDNLS},
    $$
    {\rm M}(\bm{p}_{01},\bm{q}_{11},\bm{f}_{01},\bm{g}_{01}){\rm M}(\bm{p}_{00},\bm{q}_{01},\bm{u}_{00},\bm{v}_{00})={\rm M}(\bm{p}_{10},\bm{q}_{11},\bm{u}_{10},\bm{v}_{10}) {\rm M}(\bm{p}_{00},\bm{q}_{10},\bm{f}_{00},\bm{g}_{00}),
    $$
     for any $\lambda\in\mathbb{C}$.
\end{proof}

System \eqref{ddns-dis-nc} is a noncommutative version of system (87)-(88) in \cite{SPS}. Indeed, if all $\bm{f}_{ij},\bm{g}_{ij},\bm{u}_{ij},\bm{v}_{ij},\bm{p}_{ij},\bm{q}_{ij}$, $i,j\in\{0,1\}$, then system \eqref{ddns-dis-nc} is equivalent to system (87)-(88) in \cite{SPS}. Now, since relation \eqref{ddns-dis-nc-a} is an equation solely for $\bm{f}$ and $\bm{u}$, one can eliminate these fields from the rest of equations of system \eqref{ddns-dis-nc}. 

We have the following.

\begin{corollary}
    Let $\bm{f}_{ij}, \bm{u}_{ij}\in Z(\mathcal{R})$, $i,j\in\{0,1\}$. Then, system \eqref{ddns-dis-nc} is equivalent with the system
    \begin{subequations}\label{ddns-dis-nc-simp}
        \begin{align}
        \bm{f}_{01}\bm{u}_{00}&=\bm{u}_{10}\bm{f}_{00},\\
\bm{v}_{00}+\bm{g}_{01}+\bm{p}_{01}\bm{q}_{01}&=\bm{g}_{00}+\bm{v}_{10}+\bm{p}_{10}\bm{q}_{10},\\
\bm{g}_{01}\bm{v}_{00}+\bm{p}_{01}\bm{p}_{00}+\bm{q}_{11}\bm{q}_{01}&=\bm{v}_{10}\bm{g}_{00}+\bm{p}_{10}\bm{p}_{00}+\bm{q}_{11}\bm{q}_{10},\\
\bm{q}_{01}+\bm{g}_{01}\bm{p}_{00}+\bm{p}_{01}\bm{v}_{00}&=\bm{q}_{10}+\bm{v}_{10}\bm{p}_{00}+\bm{p}_{10}\bm{g}_{00},\\
\bm{g}_{01}\bm{q}_{01}+\bm{p}_{01}+\bm{q}_{11}\bm{v}_{00}&=\bm{v}_{10}\bm{q}_{10}+\bm{p}_{10}+\bm{q}_{11}\bm{g}_{00}.
        \end{align}
    \end{subequations}
\end{corollary}

System \eqref{d2-discrete} is the commutative limit of system \eqref{ddns-dis-nc-simp} and it admits a first integral $I=(\bm{g}_{00}-\bm{p}_{00}\bm{q}_{10})(\bm{g}_{00}^2+1-\bm{p}_{00}^2-\bm{q}_{10}^2)^{-1}$. For $I=\frac{1}{2}$ in \cite{SPS} it was shown that system \eqref{d2-discrete} can be restricted to the following seven-point scalar equation
\begin{eqnarray}\label{7pointeq}
&&(w_{00}-w_{10})(w_{0,-1}-w_{-10})\left(1+\frac{1}{w_{00}-w_{1,-1}}\right)\nonumber\\
&&+(w_{00}-w_{-10})(w_{10}-w_{01})\left(1-\frac{1}{w_{00}-w_{-11}}\right)=0.
\end{eqnarray}

In the noncommutative case, system \eqref{ddns-dis-nc-simp} admits this first integral only if $\left[\bm{p}_{00},\bm{q}_{10}\right]=\left[\bm{p}_{01},\bm{q}_{11}\right]=0$,  which will not lead to a genuine noncommutative version of \eqref{ddns-dis-nc-simp}.

\section{Conclusions}\label{conclusions}
In this paper, we studied the noncommutative integrable discretisations of NLS-type systems; in particular, the AKNS nonlinear Scr\"odinger type system \eqref{NLS-eq}, the derivative NLS Kaup--Newell system \eqref{DNLS-system}, and a deformation of the NLS Mikhailov--Shabat--Yamilov system \eqref{DDNLS-sys}. 

We presented the continuum limits of NLS-type integrable lattice systems to their associated integrable PDEs. Specifically, we showed the limit from the Adler--Yamilov system \eqref{Adler-Yamilov} to the AKNS system \eqref{NLS-eq} (Proposition \ref{AY-lim}), and the continuum limit of the discrete derivative NLS system \eqref{vertex-bond-DNLS} to the Kaup--Newell system \eqref{DNLS-system} (Proposition \ref{DNLS-lim}).

Noncommutative versions of the AKNS \eqref{NLS-eq} and the Kaup--Newell \eqref{DNLS-system} already exist in the literature \cite{OS}. Here, we constructed a noncommutative analogue of the Mikhailov--Shabat--Yamilov system, namely system \eqref{system-def-DNLS} (Theorem \ref{NC-DDNLS-th}). For the noncommutative versions, we worked over an arbitrary noncommutative division ring in order to emphasise that no matrix-specific properties were assumed in the construction. However, in order to discuss their solutions we restrict ourselves to the ring of matrices $n\times n$. In fact, we demonstrated with examples that the noncommutative systems \eqref{NLS-eq-NC}, \eqref{DNLS-system-NC} and \eqref{system-def-DNLS} admit nontrivial solutions. However, in order to construct more interesting soliton solutions, one must employ the Darboux--B\"acklund transformations of such systems, which we constructed in Section \ref{NCDBT}. While the $x$-parts of the B\"acklund transformations for the NLS systems \eqref{NLS-eq-NC} and \eqref{DNLS-system-NC} were constructed in \cite{KRX} and \cite{Sokor-Nikitina}, respectively, here we presented the $t$-parts and moreover we constructed Darboux--B\"acklund transformations for system \eqref{system-def-DNLS} (Theorem \ref{DT-BT-DDNLS-th}).

Finally, we employed noncommutative Darboux--B\"acklund transformations and constructed noncommutative integrable discretisations of NLS-type. In particular, we constructed a noncommutative discrete integrable derivative NLS system \eqref{sl2-res-eq-nc} (Proposition \ref{sl2-nc-prop}) and discussed its initial value problem on the staircase. Moreover, we constructed a noncommutative 6 point equation \eqref{6-point-NC}, and discussed its associated initial value problem on the double staircase (Figure \ref{6point-ivp}). For the commutative analogue of equation \eqref{6-point-NC}, we constructed wave solutions. Furthermore, we constructed a  noncommutative discrete integrable deformation of the derivative NLS system \eqref{ddns-dis-nc} (Proposition \ref{DDNS-dis-nc-prop}).

Our results can be extended in the following ways:
\begin{itemize}
    \item Construct genuine soliton solutions for all the integrable NLS discretisations. Equation \eqref{6-point-NC} admits multi-wave solutions of the   form
$$
p(n,m) =p_0 \sum_{I\subseteq\{1,\dots,N\}} C_I\prod_{i\in I} A_i\,\rho_i^{\,n+2m}.
$$
However, these are not genuine soliton solutions because all modes depend only on 
$s=n+2m$, so the problem reduces to one dimension. The coefficients remain 
arbitrary, implying no interaction or phase shift between waves. In order to construct genuine soliton solutions, one must construct  Darboux transformations for these discretisations as in \cite{Fisenko-Sokor}.

\item Miura transformations are important tools in the theory of integrable systems. They are links between integrable equations, and are often used to deduce properties (such as conservation laws) of one equation from another. The $x$-parts of all B\"acklund transformations in this paper are certain integrable differential-difference equations. One could study their Miura transformations using recent methods from the literature \cite{Igonin}.

\item In \cite{Sokor-Sasha, Sokor-2020-2} we used Darboux transformations to construct Yang--Baxter and tetrahedron maps, and in \cite{Sokor-Sasha-2016, Sokor-Kutuzova, Sokor-Nikitina} we used noncommutative Darboux transformations to construct noncommutative 2- and 3-simplex maps. These maps are interesting themselves as discrete integrable systems and one could study their soliton solutions and also their relation to other integrable systems using recent methods \cite{Kassotakis-2023, Kassotakis-Tetrahedron, Kouloukas}.
\end{itemize}

\subsection*{Acknowledgments}
This work was supported by the Russian Science Foundation (Grant No. 26-11-00379). https://rscf.ru/project/26-11-00379/.


{\small
\begin{thebibliography}{100}
\bibitem{AKNS}
{M. J. Ablowitz, D. J. Kaup, A. C. Newell, and H. Segur.} {Nonlinear-Evolution Equations of Physical Significance. } {Phys. Rev. Lett. 31 (1973) 125}.

\bibitem{Adler-P}
{V.E. Adler.} {Painlev\'e type reductions for the non-Abelian Volterra lattices J. Phys. A: Math. Theor. 54 (2021) 035204}.

\bibitem{Adler-Sokolov}
{V.E. Adler and V.V. Sokolov.} {Non-Abelian Evolution Systems with Conservation Laws.} {Math. Phys. Anal. Geom. 24 (2021)} {7}.

\bibitem{Bobrova}
{I. Bobrova, V. Retakh, V. Rubtsov and G. Sharygin.} {Non-abelian discrete Toda chains and related lattices Physica D: Nonlinear Phenomena 464 (2024) 134200}.

\bibitem{BIKRP}
V.M. Buchstaber, S. Igonin,  S. Konstantinou-Rizos and M.M. Preobrazhenskaia. Yang--Baxter maps, Darboux transformations, and linear approximations of refactorisation problems.  J. Phys. A: Math. Theor. \textbf{53} (2020) 504002.

\bibitem{Sasha-Rhys}
{R.T. Bury and A.V. Mikhailov} 
{Automorphic Lie algebras and corresponding integrable systems} 
{Differ. Geom. Appl. \textbf{74}} {101710} (2021).


\bibitem{Fisenko-Sokor}
{X. Fisenko, S. Konstantinou-Rizos, and P. Xenitidis,} {A discrete Darboux–Lax scheme for integrable difference equations,} {Chaos, Solitons and Fractals 158 (2022)} {112059}.

\bibitem{Dimakis-Hoissen}
{A. Dimakis and F. M\"uller-Hoissen.} {Solutions of matrix NLS systems and their discretizations: a unified treatment,} {Inverse Problems 26 (2010)} {095007}.


\bibitem{Fordy-Kulish}
{A.P. Fordy and P.P. Kulish.} {Nonlinear Schr\"odinger Equations and Simple Lie Algebras,} {Commun. Math. Phys. 89 (1983)} {427--443}.

\bibitem{GGI}
V.S. Gerdjikov, G.G. Grahovski, R.I. Ivanov. On integrable wave interactions and Lax pairs on symmetric spaces, {\em Wave Motion 71} (2017) 53--70.

\bibitem{Hiet-Frank-Joshi}
J. Hietarinta, N. Joshi, F. Nijhoff (2016) Discrete Systems and Integrability, Cambridge texts in applied mathematics, Cambridge University Press.

\bibitem{Hiet-Viallet}
{J. Hietarinta and C. Viallet.} {Integrable lattice equations with vertex and bond variables} {J. Phys. A Math. Theor. 44 385201} (2011).

\bibitem{Igonin}
{S. Igonin. Matrix Lax pairs under the gauge equivalence relation induced by the gauge group action and Miura-type transformations for lattice equations.} {Journal of Geometry and Physics \textbf{216}} {105585} (2025)

\bibitem{Kassotakis-2023}
{P. Kassotakis,} {Non-Abelian hierarchies of compatible maps,
associated integrable difference systems and Yang--Baxter maps,} {Nonlinearity} {36} (2023) {2514}.

\bibitem{Kassotakis-Tetrahedron}
{P. Kassotakis, M. Nieszporski, V. Papageorgiou, and A. Tongas,} {Tetrahedron maps and symmetries of three dimensional integrable discrete equations,} {J. Math. Phys. 60 (2019)} {123503}.

\bibitem{Kaup-Newell}
{D.J. Kaup and A.C. Newell.} {An exact solution for a derivative nonlinear Schr\"odinger equation,} {Journal of Mathematical Physics 19 (1978)} {798}.

\bibitem{Sokor-2020-2}
{S. Konstantinou-Rizos,}
{Nonlinear Schr\"odinger type tetrahedron maps.} 
{\em Nuclear Physics B } {\textbf{960}} (2020) {115207}.

\bibitem{Sokor-Kutuzova}
{S. Konstantinou-Rizos and A.A. Kutuzova.}
{Noncommutative Boussinesq and NLS type 2-  and  3-simplex maps.} {Physica D: Nonlinear Phenomena 481 (2025)} {134860}

\bibitem{Sokor-Sasha-2016}
{S. Konstantinou-Rizos and A.V. Mikhailov.}
{Anticommutative extension of the Adler map.} {J. Phys. A: Math. Theor.} {\textbf{49}} (2016) {30LT03}

\bibitem{Sokor-Sasha}
{S. Konstantinou-Rizos and A.V. Mikhailov.}
{Darboux transformations, finite reduction groups and related Yang--Baxter maps.} {J. Phys. A: Math. Theor.} {\textbf{46}} (2013) {425201}

\bibitem{SPS}
{S. Konstantinou-Rizos, A.V. Mikhailov and P. Xenitidis.} 
{Reduction groups and related integrable difference systems of nonlinear Schr\"odinger type} 
{\em J. Math. Phys.} {\textbf{56}} {082701} (2015).


\bibitem{Sokor-Nikitina}
{S. Konstantinou-Rizos, A.A. Nikitina.} 
{Yang--Baxter maps of KdV, NLS and DNLS type on division rings} 
{\em Phys. D: Nonlinear Phenom.} {\textbf{465}} {134213} (2024).

\bibitem{Sokor-Pap}
{S. Konstantinou-Rizos and G. Papamikos.} 
{Entwining Yang–Baxter maps related to NLS type equations} 
{\em J. Phys. A: Math. Theor.} {\textbf{52}} {485201} (2019).


\bibitem{KRX}
{S. Konstantinou-Rizos and P. Xenitidis.} 
{Integrable discretisations of the noncommutative NLS equation} 
{J. Phys. A: Math. Theor. \textbf{59} 045203} (2026).

\bibitem{Kouloukas}
{T.E. Kouloukas. Discrete integrable systems associated with relativistic collisions. Physica D: Nonlinear Phenomena \textbf{456} 133937 (2023).}

\bibitem{Kupershmidt}
B.A. Kupershmidt. KP or MKP: Noncommutative Mathematics of Lagrangian, Hamiltonian, and Integrable Systems, {American Mathematical Society,} {ISBN: 0821814001} (2000).


\bibitem{Manakov}
S.V. Manakov. On the theory of two-dimensional stationary self focussing of electromagnetic waves {\em Sov. Phys. JETP} {\textbf{38}(2)} 248--253 (1974).


\bibitem{MPW}
{A.V. Mikhailov, G. Papamikos and J.P. Wang. Darboux Transformation for the Vector sine-Gordon Equation and Integrable Equations on a Sphere. Letters in Mathematical Physics. \textbf{106}, pages 973--996 (2016)}


\bibitem{MSY}
{A.V. Mikhailov, A.B. Shabat and R.I. Yamilov.} 
{Extension of the module of invertible transformations. Classification of integrable systems,} {Commun.Math. Phys.} {\textbf{115}} {1--19 } (1988).

\bibitem{Novikov-Wang}
{V. Novikov, J.P. Wang.} {Integrability of Nonabelian Differential–Difference Equations: The Symmetry Approach Commun. Math. Phys. 406 (2025) 11}.


\bibitem{OS}
{P.J. Olver and V. Sokolov.} 
{Non-abelian integrable systems of the derivative nonlinear Schr\"odinger type} 
{\em Inverse Problems} {\textbf{6}} {14 L5} (1998).

\bibitem{Peroni-Wang}
{E. Peroni and J. P. Wang.} 
{Darboux transformations and related non-Abelian integrable differential-difference systems of the derivative nonlinear Schr\"odinger type} 
{Physica D} {\textbf{489}} {135119} (2026).

\bibitem{Sokolov}
V. Sokolov. Algebraic Structures in Integrability. {\emph{World Scientific}} 2020, 348 pages. ISBN: 978-981-12-1965-8. D.O.I.:  10.1142/11809.

\bibitem{Xenitidis-NC}
{P. Xenitidis.} {Noncommutative discrete equations, symmetries and reductions, Physica D: Nonlinear Phenomena 483 (2025) 135004}.
\end{thebibliography}
}
\end{document}